\documentclass[twoside,11pt]{article}
\usepackage{amsmath}
\usepackage{graphicx}
\usepackage{enumerate}
\usepackage[round]{natbib} 
\usepackage{url} 

\usepackage[english]{babel} 
\usepackage{hyphenat} 
\useshorthands{~}
\defineshorthand{~-}{\hyp{}}

\usepackage{amsfonts}
\usepackage{amssymb} 
\usepackage{amsthm}
\usepackage{authblk} 
\usepackage{bm}
\usepackage{dsfont}
\usepackage{float} 
\usepackage{graphicx}
\usepackage[top=1in, bottom=1in, left=1in, right=1in]{geometry}
\usepackage{titlesec} 
\titlelabel{\thetitle.\quad} 
\usepackage{tkz-graph}  
    \usetikzlibrary{shapes.geometric,calc}

\renewcommand{\d}{\mathop{}\!\mathrm{d}}

\newcommand{\mnorm}[1]{{\left\vert\kern-0.25ex\left\vert\kern-0.25ex\left\vert #1 
    \right\vert\kern-0.25ex\right\vert\kern-0.25ex\right\vert}}

\newcommand\newday[1]{\leavevmode\xleaders\hbox{*}\hfill\kern0pt\\ \centerline{\Large \textbf{#1}}}

\usepackage[shortlabels]{enumitem}

\newtheorem{theorem}{Theorem}

\newtheorem{lemma}[theorem]{Lemma}
\newenvironment{lem}[1]
  {\renewcommand\thelemma{#1}\lemma}
  {\endlemma}

  \newtheorem{corollary}[theorem]{Corollary}

    \newtheorem{proposition}[theorem]{Proposition}

    \newtheorem{remark}{Remark}

\newtheorem{definition}{Definition}

\allowdisplaybreaks

\makeatletter
\newcommand{\pushright}[1]{\ifmeasuring@#1\else\omit\hfill$\displaystyle#1$\fi\ignorespaces}
\newcommand{\pushleft}[1]{\ifmeasuring@#1\else\omit$\displaystyle#1$\hfill\fi\ignorespaces}
\makeatother

\title{Directed Graphical Models and Causal Discovery for Zero-Inflated Data}
\author[1]{Shiqing Yu}
\author[2]{Mathias Drton}
\author[3]{Ali Shojaie}
\affil[1]{Department of Statistics, University of Washington, Seattle, Washington, 98195, U.S.A.}
\affil[2]{Department of Mathematics, Technical University of Munich, 85748 Garching bei M\"{u}nchen, Germany}
\affil[3]{Department of Biostatistics, University of Washington, Seattle, Washington, 98195, U.S.A.}

\begin{document}

\maketitle

\begin{abstract}
Modern RNA sequencing technologies provide gene expression measurements from single cells that promise refined insights on regulatory relationships among genes. Directed graphical models are well-suited to explore such (cause-effect) relationships. However, statistical analyses of single cell data are complicated by the fact that the data often show zero-inflated expression patterns. To address this challenge, we propose directed graphical models that are based on Hurdle conditional distributions parametrized in terms of polynomials in parent variables and their $0/1$ indicators of being zero or nonzero. While directed graphs for Gaussian models are only identifiable up to an equivalence class in general, we show that, under a natural and weak assumption, the exact directed acyclic graph of our zero-inflated models can be identified. We propose methods for graph recovery, apply our model to real single-cell RNA-seq data on T helper cells, and show simulated experiments that validate the identifiability and graph estimation methods in practice.
\vspace{.15cm}
\end{abstract}

\section{Introduction}
Graphical models specify conditional independence relations among variables in a random vector $\boldsymbol{Y}$ indexed by the nodes $\mathcal{V}$ of a graph $\mathcal{G}=(\mathcal{V},\mathcal{E})$ with edge set $\mathcal{E}$ \citep{handbook:2019}. Models based on undirected graphs may be used to explore conditional independence between any two variables $Y_V$ and $Y_U$ given all others $(Y_W)_{W\neq U,V}$, as represented by the absence of an edge between $V$ and $U$ in $\mathcal{E}$.   Models based on directed acyclic graphs (DAGs), for which $\mathcal{E}$ is comprised of directed edges, capture conditional independence structure that naturally arises from cause-effect relationships between the variables.

In biology and genetics, graphical models have been applied to infer the structure of gene regulatory networks based on measurements of gene expression \citep[Sections 20-21]{handbook:2019}. Traditional technologies produce expression levels aggregated over hundreds or thousands of individual cells, and these bulk measurements are frequently modeled using the assumption of Gaussianity. 
In directed Gaussian graphical models, the exact structure of the underlying DAG cannot be identified from purely observational data, and the target of inference becomes an equivalence class of DAGs.
For instance, one cannot differentiate between $V\to U$ and $U\to V$ when the variables are assumed bivariate normal.
In the Gaussian case, directed graphical models posit linear functional relationships between the variables coupled with additive Gaussian noise.
A more recent line of work emphasizes that directed graphical models that alter this assumption to nonlinear functional relationships and additive noise \citep{pet14}, or linear relations and non~-Gaussian noise \citep{shi06,sam:2020}, or linear relations with homoscedastic Gaussian noise \citep{pet13,che18} are amenable to causal discovery in the sense that different DAGs are no longer equivalent.

More recent technology obtains sequencing measurements of mRNA present in single cells. This new technology, as well as the larger sample sizes it provides, promise to give more information than bulk measurements, but at the same time bring in a unique new challenge.   At the single cell level, genes appear as ``on'' with positive single cell gene expression levels, or as ``off'' with the recorded measurements zero or negligible \citep{mcd19}. 

Figure \ref{plot_scatter} shows pairwise scatter plots of four genes from a T helper single~-cell dataset with 1951 measurements from eight healthy donors, which we analyze in Section \ref{T Helper Cell Data}. It is a superset of the single~-cell T~-follicular helper data considered in \citet{mcd19}, which is similarly plotted in their Figure 1. The lower panels show the pairwise scatterplots along with a fitted linear regression curve, and the diagonal panels show the univariate smoothed kernel density estimates for each gene. As we can see, each gene has a large number of zero values and a linear regression model is not sufficient for modeling the pairwise relationships.

\begin{figure}
\centering
\includegraphics[width=0.6\textwidth]{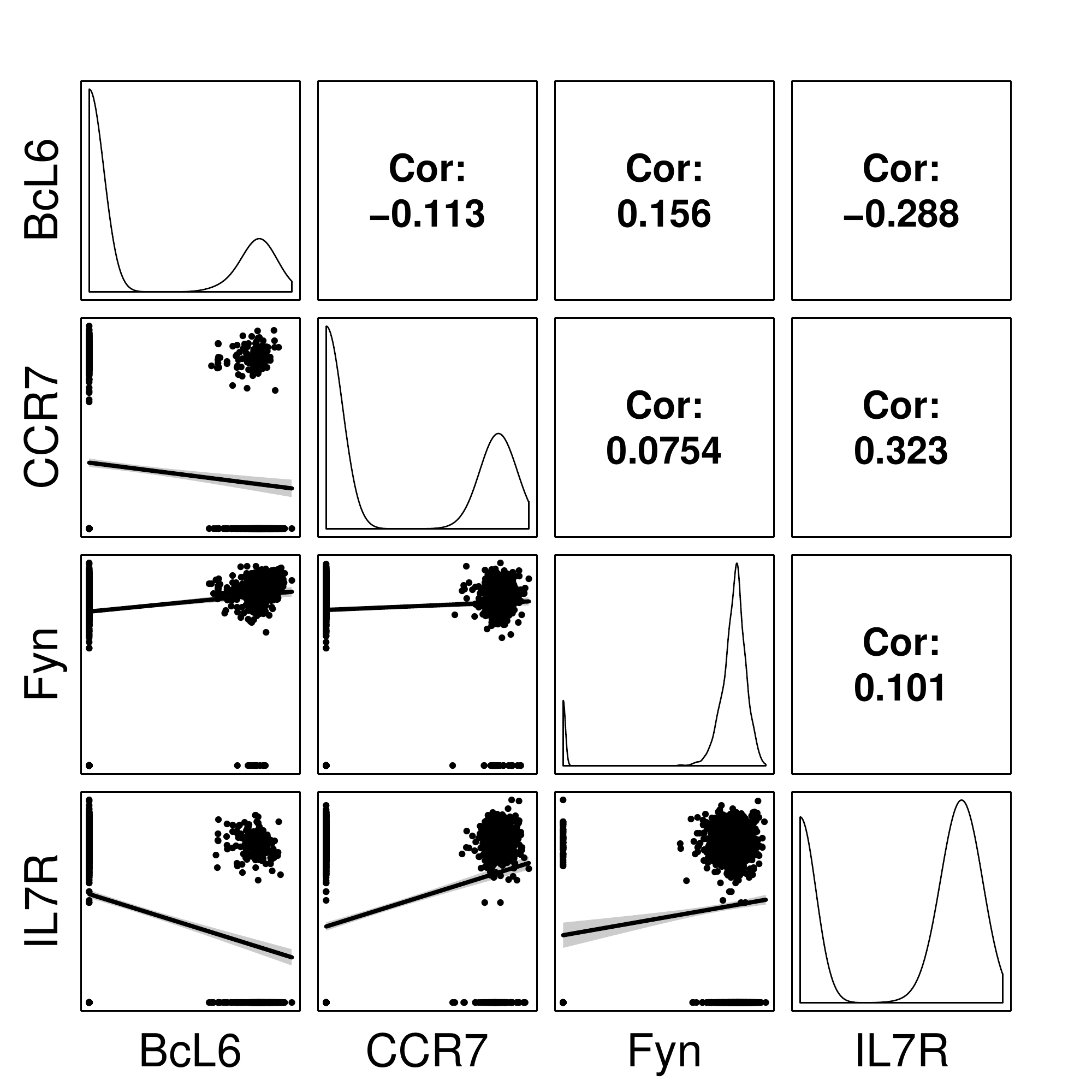}
\caption{Pairwise scatter plots and kernel densities on four genes from the T helper cell data analyzed in Section \ref{T Helper Cell Data}.}\label{plot_scatter}
\vspace{0.2in}
\end{figure}

A novel undirected graphical model that deals with this zero~-inflation was introduced by  \citet{mcd19}.  Their approach considers Hurdle density models, where for a random vector of dimension $m$, the joint probability density function has the form
\begin{equation}\label{def_undirected_joint0}
f(\boldsymbol{y};\mathbf{A},\mathbf{B},\mathbf{K})\propto \exp\left(\mathds{1}_{\boldsymbol{y}}^{\top}\mathbf{A}\mathds{1}_{\boldsymbol{y}}+\mathds{1}_{\boldsymbol{y}}^{\top}\mathbf{B}\boldsymbol{y}-\frac{1}{2}\boldsymbol{y}^{\top}\mathbf{K}\boldsymbol{y}\right),\quad\quad\boldsymbol{y}\in\mathbb{R}^m,
\end{equation}
with $\mathds{1}_{\boldsymbol{y}}$ being the elementwise indicators of nonzero entries in $\boldsymbol{y}\in\mathbb{R}^m$.  The dominating measure for this density is the $m$-fold product of the sum of the Lebesgue measure and a point mass at zero. 
However, since more information can be inferred from single~-cell sequencing data, one would hope that the data can also be analyzed using more informative directed graphical models, and that we can infer which variables (genes) are the causes of change in other variables (expression levels of other genes). 
In this paper, we formulate such directed graphical models for zero~-inflated data, and prove that under a weak assumption one can recover the exact DAG from the joint distribution.  In contrast to the setting of \citet{mcd19}, the distributions in our models are not merely zero-inflated Gaussian as we allow variables that are ``on'' to be  non-linear polynomial functions of other variables and stochastic noise.

In DAG models, the joint distribution can be factorized into the product of conditional distributions of each variable given parent variables.  For simplicity we call these conditional distributions the \emph{node conditionals}.  In our DAG model for zero~-inflated data, we form the node conditionals by taking the conditional distribution of one variable given the others in the joint model from (\ref{def_undirected_joint0}).  We refer to the resulting graphical model as the model in $(\alpha,\beta,k)$~-parametrization, or in \emph{canonical} parametrization.  Here, the $\alpha$ and $\beta$ parameters in the conditional distribution are derived from the matrix parameters $\mathbf{A}$ and $\mathbf{B}$ in (\ref{def_undirected_joint0}).  The $\alpha$ and $\beta$ are polynomials in the parent variables and their $0/1$ indicators being zero/nonzero.  An alternative second type of model is obtained by directly specifying each node conditional as a mixture of a point mass at zero and a Gaussian distribution, with the log odds of being nonzero ($\log(p/(1-p))$) and the mean in the Gaussian part being polynomials in the parent variables and their indicators.  The Gaussian variance is taken constant in the parents.  We call this second formulation the model in $(p,\mu,\sigma^2)$~- or \emph{moment} parametrization, since the parameters directly correspond to the (conditional) moments.  The detailed specification of both model types is developed in Section~\ref{DAGs for Zero-Inflated Data}.

In Section \ref{Identifiability}, we show that under our models, the distributions that can be represented by two different DAGs must be distributions of \emph{two~-Gaussian type} (Definition~\ref{def_2_gauss}).  We then prove that such distributions do not exist for dimension $m=2$ and $m=3$; we also conjecture they do not exist for $m>3$.  Moreover, we are able to prove that under a natural  and practical assumption, we have full identifiability in the sense of being able to identify the exact DAG underlying the model.  This assumption specifies that for each node, $\alpha+\beta^2/(2k)$ or equivalently $\log(p/(1-p))$ has a separate univariate term for each parent (e.g.~$y_1+y_2+y_1y_2+y_1^2$ instead of $y_1+y_1y_2+y_1^2$, which does not have a separate term for $y_2$). 

In Section~\ref{Estimation of Zero-Inflated DAGs}, we introduce different methods for estimation of the DAG. Simulation studies supporting the use of these methods are given in Section~\ref{Numerical Experiments}, and they are then applied to the T-follicular helper cell dataset (Section~\ref{T Helper Cell Data}).  


Finally, we emphasize that throughout the paper,
we use subscripts to refer to entries in vectors and columns in matrices. 
When used as a subscript of a vector, a set of nodes/indices selects the corresponding entries from the vector, e.g.,~$\boldsymbol{y}_{\mathcal{V}}=(y_V)_{V\in\mathcal{V}}$. 


\section{Directed Graphical Models for Zero-Inflated Data}\label{DAGs for Zero-Inflated Data}

In this section we motivate and formally define our models for zero-inflated data based on directed acyclic graphs (DAGs).

\subsection{Hurdle Joint Distributions for Zero-Inflated Continuous Observations}\label{sec_motivation}

\citet{mcd19} proposed 
a \emph{Hurdle joint distribution} with density
\begin{equation}\label{def_undirected_joint}
f(\boldsymbol{y};\mathbf{A},\mathbf{B},\mathbf{K})\propto \exp\left(\mathds{1}_{\boldsymbol{y}}^{\top}\mathbf{A}\mathds{1}_{\boldsymbol{y}}+\mathds{1}_{\boldsymbol{y}}^{\top}\mathbf{B}\boldsymbol{y}-\frac{1}{2}\boldsymbol{y}^{\top}\mathbf{K}\boldsymbol{y}\right),\quad\quad\boldsymbol{y}\in\mathbb{R}^m,
\end{equation}
with respect to $\lambda^m$, where $\lambda$ is the sum of a point mass at $0$ and the Lebesgue measure on $\mathbb{R}$, and $\mathbf{A}=(\alpha_{ij})_{i,j}, \mathbf{B}=(\beta_{ij})_{i,j}, \mathbf{K}=(k_{ij})_{i,j}\in\mathbb{R}^{m\times m}$ are matrices of interaction parameters with $\mathbf{K}$ positive definite.  The indicator vector $\mathds{1}_{\boldsymbol{y}}\equiv(\mathds{1}_{\{y_1\neq 0\}},\cdots,\mathds{1}_{\{y_m\neq 0\}})\in\{0,1\}^m$ captures which components of $\boldsymbol{y}$ are non-zero.

Consider a random vector $\boldsymbol{Y}\in\mathbb{R}^m$ that follows the Hurdle joint distribution. Intuitively, the density in  (\ref{def_undirected_joint}) is obtained by combining an Ising model for the indicator vector $\mathds{1}_{\boldsymbol{Y}}$ and a conditional normal distribution for $\boldsymbol{Y}$ given its nonzero pattern $\mathds{1}_{\boldsymbol{Y}}$.  The Ising model postulates a probability mass function proportional to  $\exp\left(\mathds{1}_{\boldsymbol{y}}^{\top}\mathbf{A}\mathds{1}_{\boldsymbol{y}}\right)$.  The conditional normal distribution has density 
$p\left(\boldsymbol{Y}=\boldsymbol{y}|\mathds{1}_{\boldsymbol{Y}}=\mathds{1}_{\boldsymbol{y}};\mathbf{B},\mathbf{K}\right)\propto\exp\left(\mathds{1}_{\boldsymbol{y}}^{\top}\mathbf{B}\boldsymbol{y}-\frac{1}{2}\boldsymbol{y}^{\top}\mathbf{K}\boldsymbol{y}\right)$
with respect to the Lebesgue measure restricted to the subspace of $\mathbb{R}^m$ compatible with $\mathds{1}_{\boldsymbol{y}}$. 

The exponential specification in \eqref{def_undirected_joint} entails that 
conditional independence between two variables is equivalent to the corresponding entries in all interaction matrices $\mathbf{A}$, $\mathbf{B}$, $\mathbf{K}$ being $0$.  In other words, $\alpha_{ij}=\alpha_{ji}=\beta_{ij}=\beta_{ji}=k_{ij}=k_{ji}=0$ if and only if $Y_i$ and $Y_j$ are conditionally independent given all other variables.
Indeed, it is easy to see that the induced conditional distribution of $Y_i$ given all other variables $\boldsymbol{Y}_{-i}$ in $\boldsymbol{Y}$, has density
\begin{equation}\label{eq_undirected_conditional}
p(Y_i=y_i|\boldsymbol{Y}_{-i}=\boldsymbol{y}_{-i})
=f(y_i;\alpha_{ii}+\boldsymbol{\alpha}_{i,-i}^{\top}\mathds{1}_{\boldsymbol{y}_{-i}}+\boldsymbol{\beta}_{i,-i}^{\top}\boldsymbol{y}_{-i},\beta_{ii}+\boldsymbol{\beta}_{-i,i}^{\top}\mathds{1}_{\boldsymbol{y}_{-i}}-\boldsymbol{k}_{i,-i}^{\top}\boldsymbol{y}_{-i},k_{ii}),
\end{equation}
that is, the distribution is a Hurdle distribution in $m=1$ dimension with parameters $\alpha$, $\beta$, and $k$ being linear functions in $\boldsymbol{Y}_{-i}$ and $\mathds{1}_{\boldsymbol{Y}_{-i}}$; here $f$ is the univariate version of (\ref{def_undirected_joint}). 

\subsection{Hurdle Conditionals}\label{Hurdle Conditionals}
The observation in \eqref{eq_undirected_conditional} above gives rise to the following definition. Recall that $\lambda$ is the sum of a point mass at $0$ and the Lebesgue measure on $\mathbb{R}$.

\begin{definition}[$(\alpha,\beta,k)$-Hurdle conditionals]
Given an $m$-dimensional random vector $\boldsymbol{Z}$ and a scalar random variable $X$, we say that the conditional distribution of $X$ given $\boldsymbol{Z}$ is of \emph{$(\alpha,\beta,k)$-Hurdle type} if it admits conditional densities with respect to $\lambda$ of the form
\begin{equation}\label{def_abk_hurdle_conditional}
p(X=x|\boldsymbol{Z}=\boldsymbol{z})=f_{\alpha, \beta, k}^{(m)}(X|\boldsymbol{Z}) \equiv\frac{\exp\left(\alpha(\boldsymbol{z})\mathds{1}_{x}+\beta(\boldsymbol{z})x-k x^2/2\right)}{\sqrt{2\pi/k}\exp\left(\alpha(\boldsymbol{z})+\beta^2(\boldsymbol{z})/(2k)\right)+1}.
\end{equation}
Here, $\alpha$ and $\beta$ are functions of $\boldsymbol{Z}$ (and its indicator vector). 
\end{definition}

Reparametrizing we give another intuitive formulation of Hurdle conditionals that clearly exhibits their nature of a mixture between a point mass at $0$ and a conditional Gaussian distribution.
\begin{definition}[$(p,\mu,\sigma^2)$-Hurdle conditionals]
Given an $m$-dimensional random vector $\boldsymbol{Z}$ and a scalar random variable $X$, we say that the conditional distribution of $X$ given $\boldsymbol{Z}$ is of  \emph{$(p,\mu,\sigma^2)$-Hurdle type} if it admits conditional densities with respect to $\lambda$ of the form
\begin{multline}\label{def_pms_hurdle_conditional}
p(X=x|\boldsymbol{Z}=\boldsymbol{z})=f_{p, \mu, \sigma^2}^{(m)}(X|\boldsymbol{Z})\equiv(1-p(\boldsymbol{z}))(1-\mathds{1}_{x})\\
+p(\boldsymbol{z})\mathds{1}_x\frac{1}{\sqrt{2\pi\sigma^2}}\exp\left(-\frac{(x-\mu(\boldsymbol{z}))^2}{2\sigma^2}\right).
\end{multline}
Here, $p$ and $\mu$ are functions of $\boldsymbol{Z}$ (and its indicator vector). 
\end{definition}

It is easy to show that the two parametrizations (\ref{def_abk_hurdle_conditional}) and (\ref{def_pms_hurdle_conditional}) are connected through
\begin{align}\label{eq_abk_pms_relationship}
\log\frac{p}{1-p}=\alpha+\frac{\beta^2}{2k}-\frac{1}{2}\log\left(\frac{k}{2\pi}\right),\quad\quad\mu=\frac{\beta}{k},\quad\quad\sigma^2=\frac{1}{k}.
\end{align}
That is, the conditional log odds of being nonzero is linear in $\alpha$ and quadratic in $\beta$, and the conditional Gaussian mean is proportional to $\beta$.

We note that while the $(\alpha,\beta,k)$-parametrization takes canonical parameters $\alpha(\boldsymbol{Z})$, $\beta(\boldsymbol{Z})$ and $k$ using a representation as exponential family, the moment parametrization directly models the conditional mixing probability $p(\boldsymbol{Z})$, and the mean $\mu(\boldsymbol{Z})$ and variance $\sigma^2$ parameters of the conditional Gaussian distribution. We thus refer to (\ref{def_abk_hurdle_conditional}) as the \emph{canonical parametrization}, and (\ref{def_pms_hurdle_conditional}) as the \emph{moment parametrization}.

\subsection{Directed Graphical Models for Zero-Inflation Data}

Consider an $m$-dimensional random vector $\boldsymbol{Y}$ whose components are indexed by the vertices of a DAG $\mathcal{G}=(\mathcal{V},\mathcal{E})$ and whose distribution is dominated by a product measure on $\mathbb{R}^m$.  A graphical model based on $\mathcal{G}$
requires that the density of the joint distribution admits a factorization as
\begin{equation}
    \label{eq:dag-factor}
f(\boldsymbol{y})=\prod_{V\in\mathcal{V}}f_V\left(y_V|\boldsymbol{y}_{\mathrm{pa}(V)}\right),
\end{equation}
where each factor $f_V\left(y_V|\boldsymbol{y}_{\mathrm{pa}(V)}\right)$ is a conditional density for $y_V$ given its parent variables $\boldsymbol{y}_{\mathrm{pa}(V)}$.  The set of parents is defined to be  $\mathrm{pa}(V)\equiv\{U:U\to V\in\mathcal{E}\}$.

In Section \ref{sec_motivation} we observed that, for the Hurdle joint distributions from (\ref{def_undirected_joint}), the conditional distribution of one variable $Y_i$ given the others is an $(\alpha,\beta,k)$-Hurdle with $k$ constant, and $\alpha$ and $\beta$ linear functions of those variables (and their indicators) that are conditionally dependent on $Y_i$; see  \eqref{eq_undirected_conditional}. Motivated by this fact, we specify directed graphical models for zero-inflated data by assuming the conditional densities in the factorization in \eqref{eq:dag-factor} to be $(\alpha,\beta,k)$-  or $(p,\mu,\sigma^2)$-Hurdle conditionals.  We then assume the parameters in these conditionals to be \emph{Hurdle polynomials} in its parents, as defined now.

\begin{definition}[Hurdle polynomials]
Let $\boldsymbol{Y}=(Y_V)_{V\in\mathcal{V}}$ be an $m$-dimensional random vector indexed by a set $\mathcal{V}$, and suppose $\mathcal{S}\subseteq\mathcal{V}$. If $\mathcal{S}\neq\varnothing$, define the space of \emph{Hurdle polynomials} in $\boldsymbol{y}_\mathcal{S}$ as
\begin{multline}\label{def_polynomial}
\mathcal{H}(\boldsymbol{Y};\mathcal{S})\equiv\left\{c_0+\sum_{j=1}^{T}c_{j}\prod_{U\in \mathcal{U}_{j}}Y_U^{d_{j,U}}\prod_{V\in \mathcal{V}_{j}} \mathds{1}_{Y_V},\quad c_0\in\mathbb{R},\,T\in\mathbb{N},\right.\\
\left.\phantom{\sum_{j=1}^{T}}c_j\neq 0,\,\mathcal{U}_j\subseteq \mathcal{S},\,\mathcal{V}_j\subseteq \mathcal{S}\backslash \mathcal{U}_j,\,d_{j,U}\in\mathbb{N}\;\quad\forall U\in \mathcal{U}_j\quad\forall j=1,\,\dots,\,T\right\},
\end{multline}
where $\mathbb{N}=\{1,2,\dots\}$.
This is the set of polynomials in values and indicators of nodes in $\mathcal{S}$. If $\mathcal{S}=\varnothing$, define $\mathcal{H}(\boldsymbol{Y};\mathcal{S})\equiv\mathbb{R}$.  The \emph{degree} of a hurdle polynomial as specified in \eqref{def_polynomial} is $\max\limits_{j=1,\dots,T}\sum\limits_{U\in \mathcal{U}_j}d_{j,U}+|\mathcal{V}_j|$. Here $|\cdot|$ denotes the set cardinality. 
\end{definition}
We are now ready to formally define our models.
\begin{definition}[DAG models for zero-inflated data]\label{def_zero_inflated_dags}
Let $\mathcal{G}=(\mathcal{V},\mathcal{E})$ be a DAG with $|\mathcal{V}|=m$ nodes.
A zero-inflated conditional Gaussian DAG model associated with  $\mathcal{G}$ is a set of joint distributions on $\mathbb{R}^m$ that admit a density (with respect to $\lambda^m$) that factors as in \eqref{eq:dag-factor} with 
each conditional density $f_V\left(y_V|\boldsymbol{y}_{\mathrm{pa}(V)}\right)$ being a Hurdle conditional 
\begin{enumerate}[(1)]
\item in the $(\alpha,\beta,k)$-parametrization with parameters $\alpha_V$, $\beta_V$ and $k_V$, where $k_V$ is constant, $\alpha_V$ and $\beta_V$ are Hurdle polynomials in $\boldsymbol{y}_{\mathrm{pa}(V)}$; or
\item in the $(p,\mu,\sigma^2)$-parametrization with parameters $p_V$, $\mu_V$ and $\sigma_V^2$, where $\sigma_V^2$ is constant, $\log(p_V/(1-p_V))$ and $\mu_V$ are Hurdle polynomials in $\boldsymbol{y}_{\mathrm{pa}(V)}$.
\end{enumerate}
\end{definition}
It is apparent from the relationship (\ref{eq_abk_pms_relationship}) that if we allow the relevant parameters to be Hurdle polynomials of \emph{any} degree, the two parametrizations are equivalent, meaning that given an underlying DAG, they share the same space of all possible joint distributions. However for computational convenience it is useful to bound the degree. In later applications, we will only consider degrees up to three.


\section{Identifiability}\label{Identifiability}

\subsection{Strong Identifiability}

As we show next, the directed graphical models from Definition~\ref{def_zero_inflated_dags} are amenable to causal discovery in the sense that the DAG underlying the model is uniquely identifiable from a given joint distribution.  More precisely, we prove identifiability under an explicit mild assumption on the Hurdle conditionals determining the considered joint distribution.


Let $\pi(\boldsymbol{y}_\mathcal{S})\in \mathcal{H}(\boldsymbol{Y};\mathcal{S})$ be a Hurdle polynomial for a subset $\mathcal{S}\subseteq\mathcal{V}$. For $U\in \mathcal{S}$, let $\pi_U(y_U)\equiv \pi(y_U,\boldsymbol{0})$ be the restriction of $\pi(\boldsymbol{y}_\mathcal{S})$ obtained by setting all entries other than $y_U$ to zero.
Then $\pi_U(y_U)\in\mathcal{H}(\boldsymbol{Y};\{U\})$ is a 
univariate Hurdle polynomial.

\begin{definition}[Strong Hurdle polynomials]
Let $\pi(\boldsymbol{y}_\mathcal{S})\in \mathcal{H}(\boldsymbol{Y};\mathcal{S})$. We say
$\pi(\boldsymbol{y}_\mathcal{S})$ is a \emph{strong Hurdle polynomial} if all of its restrictions $\pi_U(y_U)$ take at least three different values. 
In other words, for each $U\in\mathcal{S}$, the Hurdle polynomial $\pi(\boldsymbol{y}_\mathcal{S})$ contains at least one 
term that depends only on $(y_U,\mathds{1}_{Y_U})$ and is of the form $c_j y_U^d\mathds{1}_{Y_U}$ or $c_j y_U^d$ with $c_j\not=0$ and $d\ge 1$.
\end{definition}

\begin{theorem}[DAG identifiability with strong Hurdle polynomials]
\label{thm_full_identifiability}

Let $f(\boldsymbol{y})$ be a joint density with respect to $\lambda^m$ that factors according to a DAG $\mathcal{G}=(\mathcal{V},\mathcal{E})$, as in \eqref{eq:dag-factor}.   Suppose for each $V\in\mathcal{V}$, the conditional $f_V\left(y_V|\boldsymbol{y}_{\mathrm{pa}(V)}\right)$ is of Hurdle type with parameters $(\alpha_V,\beta_V,k_V)$ or $(p_V,\mu_V,\sigma^2_V)$.  If for each $V$, $\alpha_V+\beta_V^2/(2k_V)$, or equivalently $\log (p_V/(1-p_V))$, is a strong Hurdle polynomial, then $f(\boldsymbol{y})$ does not factor with respect to any other DAG $\mathcal{G}'\not=\mathcal{G}$.
\end{theorem}

In the proof  in the \ref{appendix_proofs} we show that the given restriction on the parameters of the Hurdle conditionals is actually stronger than what we need for identifiability. However, the assumption of \emph{strong} Hurdle polynomials is very natural in that it specifies a weak form of hierarchy among interactions by requiring that the conditional distributions are parametrized to include at least one univariate power term in every parent variable 
and not just indicators or interaction terms with other parents.

\subsection{Weak Identifiability}

Without assuming the Hurdle polynomials for the conditional distributions to be \emph{strong}, we can still offer a weaker identifiability result that shows that the distributions in the intersection between the models obtained from two Markov equivalent DAGs with Hurdle polynomial parameters 
always have to be of what we call \emph{two-Gaussian type}.
In our definition of this concept, we write $\phi(\,\cdot\,;\mu,\nu)$ for the univariate normal density function for mean $\mu$ and inverse variance $\nu$.

\begin{definition}\label{def_2_gauss}
Let $\boldsymbol{Y}=(Y_V)_{V\in\mathcal{V}}$ be a random vector, and let $W,\,U\in\mathcal{V}$ be the indices for two of its components.  Further, let $\mathcal{P}\subseteq\mathcal{V}\backslash\{W,U\}$ be a set of additional indices.  Then the joint distribution of $\boldsymbol{Y}$ is of \emph{two-Gaussian type w.r.t.~$(W,U,\mathcal{P})$} if the following holds for both $V=W$ and $V=U$:  There exists a constant $\nu^V_{1}$, polynomials $\mu^V_{1}(\boldsymbol{y}_\mathcal{P})$, $\mu^V_{2}(\boldsymbol{y}_\mathcal{P})$, $\nu^V_{2}(\boldsymbol{y}_\mathcal{P})$, and functions $c^V_{1}(\boldsymbol{y}_\mathcal{P})$ and $c^V_{2}(\boldsymbol{y}_\mathcal{P})$ such that for almost every $\boldsymbol{y}_\mathcal{P}\in\mathbb{R}^{|\mathcal{P}|}$, $c_{1}^V(\boldsymbol{y}_{\mathcal{P}})> 0$, $c_{2}^V(\boldsymbol{y}_{\mathcal{P}})> 0$, either $\mu^V_{1}(\boldsymbol{y}_{\mathcal{P}})\neq\mu^V_{2}(\boldsymbol{y}_{\mathcal{P}})$ or $\nu^V_{1}\neq \nu^V_{2}(\boldsymbol{y}_{\mathcal{P}})$, and the conditional density 
\begin{multline*}
\mathbb{P}\left(Y_{V}=y\left|Y_{V}\neq 0,\boldsymbol{Y}_\mathcal{P}=\boldsymbol{y}_\mathcal{P}\right.\right)=c^V_{1}(\boldsymbol{y}_\mathcal{P})\phi\left(y;\mu^V_{1}(\boldsymbol{y}_\mathcal{P}),\nu^V_{1}\right)
+c^V_{2}(\boldsymbol{y}_\mathcal{P})\phi\left(y;\mu^V_{2}(\boldsymbol{y}_\mathcal{P}),\nu^V_{2}(\boldsymbol{y}_\mathcal{P})\right),
\end{multline*}
is a mixture of exactly two distinct Gaussian distributions with means polynomial in $\boldsymbol{y}_\mathcal{P}$, one with an absolute constant inverse variance parameter and the other polynomial in $\boldsymbol{y}_\mathcal{P}$.

If $\mathcal{P}=\varnothing$, then 
two-Gaussian type w.r.t.~$(W,U,\varnothing)$ requires that
both $\mathbb{P}(Y_{W}|Y_{W}\neq 0)$ and $\mathbb{P}(Y_{U}|Y_{U}\neq 0)$ are mixtures of exactly two distinct univariate Gaussian distributions with constant parameters, respectively.
\end{definition}

We next recall the following observation that appears as Proposition 29(ii) in \citet{pet14}; see Sections 1.8 and 15.3.2 of \citet{handbook:2019} for definitions of the Markov property and faithfulness. 

\begin{proposition}
\label{prop_peters}
Suppose the distribution of $\boldsymbol{Y}$ is Markov and faithful with respect to two distinct Markov equivalent graphs $\mathcal{G}$ and $\mathcal{G}'$. Then, there must exist nodes $W$ and $U$ such that $W\to U$ in $\mathcal{G}$ and $U\to W$ in $\mathcal{G}'$, while $\mathcal{P}\equiv\mathrm{pa}_{\mathcal{G}}(U)\backslash\{W\}=\mathrm{pa}_{\mathcal{G}'}(W)\backslash\{U\}$.
\end{proposition}

\begin{remark}
Proposition \ref{prop_peters} is at the heart of many proofs of DAG identifiability, which combine it with suitable probabilistic conditioning to reduce the comparison of two DAG models to bivariate problems involving the two graphs $W\to U$ and $W\leftarrow U$.  However, in our setting, a key new challenge arises because the form of the Hurdle conditionals precludes us from applying conditioning to form sets of bivariate distributions that are of the considered Hurdle type. Indeed, conditioning on descendants of the considered variables (i.e., other variables that in the graph can be reached along directed paths) generally gives conditional distributions that are no longer of the Hurdle type used in the definition of our model class.
\end{remark}


We claim that the intersection of sets of 
joint distributions represented by two distinct Markov equivalent $\mathcal{G}$ and $\mathcal{G}'$ must be a subset of $2$-Gaussian type distributions with respect to a triplet $(W,U,\mathcal{P})$ obtained from Proposition \ref{prop_peters}.

\begin{theorem}[General Identifiability]\label{id_full}
Let $\boldsymbol{Y}$, $\mathcal{G}$, $\mathcal{G}'$, $W$, $U$, $\mathcal{P}$ be as in Proposition \ref{prop_peters}.  Let $\boldsymbol{Y}$ have a $\lambda^m$-density that factors w.r.t.~both graphs $\mathcal{G}$ and $\mathcal{G}'$.  For each $\mathcal{H}=\mathcal{G},\mathcal{G}'$, let the node conditionals in the factorization be Hurdle conditionals with the parameters $(\alpha_V^{\mathcal{H}})_{V\in \mathcal{V}}$ and $(\beta_V^{\mathcal{H}})_{V\in \mathcal{V}}$ from (\ref{def_abk_hurdle_conditional}), or equivalently $(p_V^{\mathcal{H}})_{V\in \mathcal{V}} $ and $(\mu_V^{\mathcal{H}})_{V\in \mathcal{V}}$ from (\ref{def_pms_hurdle_conditional}), that are Hurdle polynomials of the form (\ref{def_polynomial}), 
where for $(V,\,T,\,\mathcal{H})=(U,\,W,\,\mathcal{G})$ and $(V,\,T,\,\mathcal{H})=(W,\,U,\,\mathcal{G}')$ it holds that
\begin{enumerate}[(i)]
\item $\beta_V^{\mathcal{H}}(y_T,\boldsymbol{y}_{\mathcal{P}})$ (or $\mu_V^{\mathcal{H}}(y_T,\boldsymbol{y}_{\mathcal{P}})$) depends on at least one of $\mathds{1}_{y_T}$ and $y_T$, or 
\item $\alpha_V^{\mathcal{H}}(y_T,\boldsymbol{y}_{\mathcal{P}})$ (or $p_V^{\mathcal{H}}(y_T,\boldsymbol{y}_{\mathcal{P}})$) depends on the value of $y_T$ (and maybe additionally on $\mathds{1}_{y_T}$).
\end{enumerate} 
Then the distribution of  $\boldsymbol{Y}$ must be of two-Gaussian type w.r.t.~$(W,U,\mathcal{P})$. In this case we also say the distribution is of two-Gaussian type w.r.t.~$\mathcal{G}$ and $\mathcal{G}'$.
\end{theorem}
Note that the assumption of faithfulness in Proposition \ref{prop_peters} implies that we have (i) or (ii) or a condition (iii) that states that $\alpha_V^{\mathcal{H}}(y_T,\boldsymbol{y}_{\mathcal{P}})$ (or $p_V^{\mathcal{H}}(y_T,\boldsymbol{y}_{\mathcal{P}})$) depends on $\mathds{1}_{y_T}$ only and $\beta_V^{\mathcal{H}}(y_T,\boldsymbol{y}_{\mathcal{P}})$ (or $\mu_V^{\mathcal{H}}(y_T,\boldsymbol{y}_{\mathcal{P}})$) is constant in $y_T$. It is case (iii) that we rule out in our assumption of Theorem \ref{id_full}.


The result is proved in the \ref{appendix_proofs}. It is easy to show that the result also holds if we make modifications such as restricting the maximum degree of the polynomial or excluding interactions between the discrete and continuous components.

In the two- and three-dimensional cases (i.e., $m=2,3$) we show in the \ref{appendix_proofs} that there does not exist a joint distribution for $\boldsymbol{Y}$ that is of two-Gaussian type with respect to two distinct Markov equivalent graphs. We thus have the following result on full identifiability for graphs with two or three nodes.

\begin{corollary}[Identifiability in two and three dimensions]\label{id_binary}
If $|\mathcal{V}|\leq 3$, i.e., in a binary/triary setting, there does not exist a joint distribution that is of two-Gaussian type w.r.t.~two distinct Markov equivalent DAGs $\mathcal{G}$ and $\mathcal{G}'$. Thus, strong identifiability is guaranteed as in Theorem \ref{thm_full_identifiability}, meaning that the sets of Markov and faithful distributions associated to $\mathcal{G}$ and $\mathcal{G}'$ must be disjoint.
\end{corollary}

Theorem \ref{thm_full_identifiability} and Corollary \ref{id_binary} state that the DAGs are perfectly identifiable from the distributions if $m=2,3$ or if we assume the Hurdle polynomials to be \emph{strong}; Theorem \ref{id_full} claims that without assuming \emph{strong} Hurdle polynomials, the distributions for $m>3$ from which the graph is not identifiable must be a subset of the \emph{two-Gaussian type} distributions. We conjecture that in general, with $m>3$, the set of two-Gaussian type distributions with respect to any two graphs is an empty set. 

In Figure \ref{plot_identifiability} we show scatter plots of simulated data that give some indication of how Markov equivalent graphs may be differentiated under our models.

\section{Estimation of DAGs from Zero-Inflated Data}\label{Estimation of Zero-Inflated DAGs}

Suppose now that we are given an i.i.d.~sample $\boldsymbol{y}^{(1)},\ldots,\boldsymbol{y}^{(n)}$ comprised of $m$-variate observations.  
The log-likelihood function $\ell$ of any DAG model can be decomposed into the sum of conditional (or nodewise) log-likelihood functions $\ell^V$ for the $V$-th variable conditional on its parent variables.  Let $y_V^{(1)},\ldots,y_V^{(n)}$ be the $n$ observations of the $V$-th variable.  For the canonical $(\alpha,\beta,k)$-parametrization from (\ref{def_abk_hurdle_conditional}), the nodewise log-likelihood function is  
\begin{multline*}
\ell^{V}\left(\alpha_V,\beta_V,k_V\left|\boldsymbol{y}^{(1)},\ldots,\boldsymbol{y}^{(n)}\right.\right)=\sum_{i=1}^n\left(\alpha_V\left(\boldsymbol{y}_{\mathrm{pa}(V)}^{(i)}\right)\mathds{1}_{y_V^{(i)}}+\beta_V\left(\boldsymbol{y}_{\mathrm{pa}(V)}^{(i)}\right)y_V^{(i)}-k_V{y_V^{(i)}}^2/2\right.\\
-\log\left.\left[\sqrt{2\pi/k_V}\exp\left\{\alpha_v\left(\boldsymbol{y}_{\mathrm{pa}(V)}^{(i)}\right)+\beta_V^2\left(\boldsymbol{y}_{\mathrm{pa}(V)}^{(i)}\right)/(2k_V)\right\}+1\right]\right);
\end{multline*}
for the moment $(p,\mu,\sigma^2)$-parametrization from (\ref{def_pms_hurdle_conditional}) it is
\begin{multline*}
\ell^{V}\left(p_V,\mu_V,\sigma_V^2\left|\boldsymbol{y}^{(1)},\ldots,\boldsymbol{y}^{(n)}\right.\right)=\sum_{i:y_V^{(i)}=0}\log\left\{1-p_V\left(\boldsymbol{y}_{\mathrm{pa}(V)}^{(i)}\right)\right\}\\
+\sum_{i:y_V^{(i)}\neq 0}\left[\log p_V\left(\boldsymbol{y}_{\mathrm{pa}(V)}^{(i)}\right)-\frac{1}{2}\log(2\pi\sigma_V^2)-\left\{y_V^{(i)}-\mu_V\left(\boldsymbol{y}_{\mathrm{pa}(V)}^{(i)}\right)\right\}^2/(2\sigma_V^2)\right].
\end{multline*}
In the latter case, we see the sum of the log-likelihood functions from the logistic regression model for $p_V$ and the linear regression for $\mu_V$ restricted to the observations with $y_V^{(i)}\neq 0$. Here we recall that the parameters $\alpha_V, \beta_V, p_V, \mu_V$ are themselves polynomials in $\boldsymbol{y}_{\mathrm{pa}(V)}$ and their indicators, and we are using them as a shorthand notation on the left-hand sides where we really mean $\ell^V$ as a function of the parameters (i.e., coefficients) in those polynomials.

\subsection{Fitting Hurdle Conditionals}\label{Fitting Hurdle Conditionals}

Estimation of the graphical models amounts to fitting the conditional distribution of one node given a set of others. For the canonical $(\alpha,\beta,k)$-parametrization, the log-likelihood function is convex in $\alpha_V$, $\beta_V$ and $k_V$. Moreover, $\alpha_V$ and $\beta_V$ are linear in the polynomial coefficients.  Therefore, the log-likelihood is convex in the coefficients to estimate and can be maximized by standard methods; e.g., coordinate descent. Estimation for the moment $(p,\mu,\sigma^2)$-parametrization (\ref{def_pms_hurdle_conditional}), on the other hand, can be easily solved by separately fitting a logistic regression to $p_V$ and a linear regression to $\mu_V$. 
Recall again that the two parametrizations, canonical and moment, are equivalent when assuming a full polynomial model, i.e., when the degree and structure of the polynomials is unrestricted.  However, when restricting, for instance, the degree the two parametrizations yield different models.

The $(\alpha,\beta,k)$-parametrization with linear Hurdle polynomials (i.e., degree $1$) is interesting as it naturally comes from conditional distributions of the joint distribution defined for undirected graphical models in \citet{mcd19}.  However, at least for higher degree, the $(p,\mu,\sigma^2)$-parametrization may be more intuitive and useful in practice as it leads to a decomposition into a logistic regression and a linear regression.  This decomposition enables one to use optimized standard regression solvers for model fitting. The $(p,\mu,\sigma^2)$-parametrization also makes it easy to apply available routines to  incorporate regularization on the coefficients/parameters into our loss, which is helpful when the number of samples is small compared to the number of parameters.  Such higher dimensionality of the models arises in particular when assuming a higher degree for the Hurdle polynomials.  The regularization is automatically applied in the implementation in our R package \texttt{ZiDAG} available on \texttt{GitHub}.

For estimation of our models, we assume a highest degree of the Hurdle polynomials.  To select the degree from data we adopt the Bayesian information criterion (BIC).  This functionality is incorporated in \texttt{ZiDAG}.

\subsection{Graph Search}\label{Graph Search}

For estimation of the DAG underlying the graphical model, we mainly consider two state-of-the-art methods: (A) exhaustive score-based search and (B) 
greedy search.
Both methods rely on a model score which we take to be the BIC defined as $\nu\log n-2\ell$, where $\nu$ is the total number of parameters in the model, $n$ is the sample size, and $\ell$ is the log-likelihood as introduced in the beginning of Section \ref{Estimation of Zero-Inflated DAGs}.

\begin{enumerate}[(A)]
\item \textbf{Exhaustive search}:  Optimizing the BIC over the set of all DAGs is possible for moderately small $m$ using the dynamic programming algorithm of \citet{sil06}.  This approach is justified by the asymptotic consistency of the BIC as well as the identifiability of our model (recall Section \ref{Identifiability}). The experiments of \citet{sil06} suggest that for Gaussian models the search is practical for $m<32$. Estimation of our models is computationally more challenging but exhaustive search is feasible at least for $m<16$.  
\item \textbf{Greedy search}:  Instead of optimizing BIC over all DAGs, we may apply a greedy search that iteratively improves BIC by moving to a neighboring DAG that provides the largest improvement.  The neighborhood is defined using edge additions, deletions, and reversals; compare \citet{chi02}.
While \citet{chi02} discusses consistency of graph recovery in terms of equivalence classes, in our case the algorithm determines individual graphs.
For faster estimation in sparse settings, we consider restricting the maximum node degree (i.e., the maximum number of parents).
\end{enumerate}

\begin{remark}
We have also experimented with a version of the PC algorithm, 
which is not easily applicable since it relies on a suitable conditional independence test between two variables given any potential parent set.  Indeed, by the nature of our models if the potential parent set is misspecified, the conditional distributions may no longer be Hurdle. 
Another possible approach starts with greedily estimating the topological ordering of nodes by iteratively picking the node that maximizes the conditional likelihood given nodes already chosen, followed by a variable selection problem using, for example, a Wald test or $\ell_1$ regularization techniques; this method relies on very subtle features of the distributions.  Neither the PC algorithm we designed nor the approach focusing on the topological ordering were competitive in our experiments.
\end{remark}

\subsection{Stability Selection}\label{Stability Selection}

In our application to single-cell gene expression data, we seek to also achieve some control of the false discovery rate (FDR).  To this end, we apply stability selection in graph estimation. In particular, we take up the approach outlined in \citet{sha13}. We randomly choose $B=50$ subsets of the data (each of size $\lfloor n/2\rfloor$), and obtain $B$ other sets as their complements of equal size, randomly throwing out one sample if $n$ is odd.  We then estimate the graph using $\lfloor n/2\rfloor$ subsamples indexed by each of these $2B$ sets of equal size, and obtain $2B$ estimated DAGs. Given the desired FDR, we compute a frequency threshold using the formula from \citet[Eqn.~(8)]{sha13} with number of total parameters $m(m-1)/2$. We then keep all edges that occur more often than the frequency threshold and produce a graph as our final estimate. In our implementation in \texttt{ZiDAG}, if a graph estimated this way is not acyclic, the user can choose to return it as is, or the function will increase the threshold up to the point where the resulting graph is a DAG, even though the resulting graph might be empty in extreme cases.

\section{Numerical Experiments}\label{Numerical Experiments}

In this section we present numerical experiments for exact DAG recovery using simulated zero-inflated conditional Gaussian data. The main goal is to verify identifiability using exhaustive search, and examine how accurately greedy search can recover the true graph.

\subsection{True Underlying DAGs and Distributions}\label{True Underlying DAGs and Distributions}
We consider three DAG structures, i) chain graph with $m=5$, ii) complete graph with $m=5$, iii) lattice graph with $m=9$, as illustrated in Figure~\ref{plot_true_DAGs}. We keep $m$ small for tractability of repeated application of \citet{sil06} with quadratic Hurdle polynomials. In practice, we suggest applying the directed graph models to the connected components inferred from undirected graphs, estimated using the joint Hurdle distribution (\ref{def_undirected_joint}) as in \citet{mcd19}. This often results in considerably smaller sizes $m$. In particular, sizes of $m$ in this section are similar to the largest component in our data analysis in Section \ref{T Helper Cell Data}.

\begin{figure}[t!]
\begin{center}
\begin{tikzpicture}[baseline,scale=1]
    \node[shape=circle,draw=black, line width=0.5mm] (1) at (0, 1){};
    \node[shape=circle,draw=black, line width=0.5mm] (2) at ($({sqrt(10+2*sqrt(5))/4},{(sqrt(5)-1)/4})$) {};
    \node[shape=circle,draw=black, line width=0.5mm] (3) at ($({sqrt(10-2*sqrt(5))/4},{-(sqrt(5)+1)/4})$) {} ;
      \node[shape=circle,draw=black, line width=0.5mm] (4) at ($({-sqrt(10-2*sqrt(5))/4},{-(sqrt(5)+1)/4})$) {};
    \node[shape=circle,draw=black, line width=0.5mm] (5) at ($({-sqrt(10+2*sqrt(5))/4},{(sqrt(5)-1)/4})$) {};
    \path [->, line width=0.5mm](1) edge node[right] {} (2);
    \path [->, line width=0.5mm](2) edge node[right] {} (3);
    \path [->, line width=0.5mm](3) edge node[right] {} (4);
    \path [->, line width=0.5mm](4) edge node[right] {} (5);
    \node[below=2cm] at (current bounding box.base) {Chain, $m=5$};
\end{tikzpicture}
\hfill
\begin{tikzpicture}[grow=right,baseline,scale=1]
    \node[shape=circle,draw=black, line width=0.5mm] (1) at (0, 1){};
    \node[shape=circle,draw=black, line width=0.5mm] (2) at ($({sqrt(10+2*sqrt(5))/4},{(sqrt(5)-1)/4})$) {};
    \node[shape=circle,draw=black, line width=0.5mm] (3) at ($({sqrt(10-2*sqrt(5))/4},{-(sqrt(5)+1)/4})$) {} ;
      \node[shape=circle,draw=black, line width=0.5mm] (4) at ($({-sqrt(10-2*sqrt(5))/4},{-(sqrt(5)+1)/4})$) {};
    \node[shape=circle,draw=black, line width=0.5mm] (5) at ($({-sqrt(10+2*sqrt(5))/4},{(sqrt(5)-1)/4})$) {};
    \path [->, line width=0.5mm](1) edge node[right] {} (2);
    \path [->, line width=0.5mm](1) edge node[right] {} (3);
    \path [->, line width=0.5mm](1) edge node[right] {} (4);
    \path [->, line width=0.5mm](1) edge node[right] {} (5);
    \path [->, line width=0.5mm](2) edge node[right] {} (3);
    \path [->, line width=0.5mm](2) edge node[right] {} (4);
    \path [->, line width=0.5mm](2) edge node[right] {} (5);
    \path [->, line width=0.5mm](3) edge node[right] {} (4);
    \path [->, line width=0.5mm](3) edge node[right] {} (5);
    \path [->, line width=0.5mm](4) edge node[right] {} (5);
    \node[below=2cm] at (current bounding box.base) {Complete, $m=5$};
\end{tikzpicture}
\hfill
\begin{tikzpicture}[grow=right,baseline,scale=1]
    \node[shape=circle,draw=black, line width=0.5mm] (1) at (0, 0.9){};
    \node[shape=circle,draw=black, line width=0.5mm] (2) at (0.9, 0.9) {};
    \node[shape=circle,draw=black, line width=0.5mm] (3) at (1.8, 0.9) {} ;
    \node[shape=circle,draw=black, line width=0.5mm] (4) at (0, 0) {};
    \node[shape=circle,draw=black, line width=0.5mm] (5) at (0.9, 0) {};
    \node[shape=circle,draw=black, line width=0.5mm] (6) at (1.8, 0) {};
    \node[shape=circle,draw=black, line width=0.5mm] (7) at (0, -0.9) {};
    \node[shape=circle,draw=black, line width=0.5mm] (8) at (0.9, -0.9) {};
    \node[shape=circle,draw=black, line width=0.5mm] (9) at (1.8, -0.9) {};
    \path [->, line width=0.5mm](1) edge node[right] {} (2);
    \path [->, line width=0.5mm](1) edge node[right] {} (4);
    \path [->, line width=0.5mm](2) edge node[right] {} (3);
    \path [->, line width=0.5mm](2) edge node[right] {} (5);
    \path [->, line width=0.5mm](3) edge node[right] {} (6);
    \path [->, line width=0.5mm](4) edge node[right] {} (5);
    \path [->, line width=0.5mm](4) edge node[right] {} (7);
    \path [->, line width=0.5mm](5) edge node[right] {} (6);
    \path [->, line width=0.5mm](5) edge node[right] {} (8);
    \path [->, line width=0.5mm](6) edge node[right] {} (9);
    \path [->, line width=0.5mm](7) edge node[right] {} (8);
    \path [->, line width=0.5mm](8) edge node[right] {} (9);
    \node[below=2cm] at (current bounding box.base) {Lattice, $m=9$};
\end{tikzpicture}
\end{center}
\caption{Graph structures used in our experiments.}\label{plot_true_DAGs}
\vspace{0.2in}
\end{figure}

For each structure, we consider true generating conditional distributions using the following parametrizations: a) $(\alpha,\beta,k)$-(canonical) parametrization with \emph{linear} Hurdle polynomials, b) $(p,\mu,\sigma^2)$-(moment) parametrization with \emph{linear} Hurdle polynomials, and c) $(p,\mu,\sigma^2)$-(moment) parametrization with \emph{quadratic} Hurdle polynomials. We note that the distributions represented by c) is a superset of those by a) and b). By (\ref{eq_abk_pms_relationship}), distributions represented by a) and b) are disjoint because $\log(p/(1-p))$ is a weighted sum of $\alpha$ and $\beta^2$.

Recall the definition of Hurdle conditionals in (\ref{def_abk_hurdle_conditional}) and (\ref{def_pms_hurdle_conditional}) in Section \ref{Hurdle Conditionals}. In our experiments, whenever $\mathrm{pa}(V)=\varnothing$, we generate $y_V\sim f_0$ such that $f_0(x)=\frac{1}{2}(1-\mathds{1}_x)+\frac{1}{2}\phi(x;0,1)$, where $\phi$ is the standard normal density. Otherwise, for parametrization a), we use Hurdle conditionals with parameters $k_V=1$,
$\alpha_V(\boldsymbol{y}_{\mathrm{pa}(V)})=\sum_{U\in \mathrm{pa}(V)}\mathds{1}_{y_U}+y_U$ and $\beta_V(\boldsymbol{y}_{\mathrm{pa}(V)})=\sum_{U\in \mathrm{pa}(V)}\left(\mathds{1}_{y_U}-y_U\right)$; similarly for parametrization b) we take $\sigma^2_V=1$,
$\log\frac{p_V}{1-p_V}(\boldsymbol{y}_{\mathrm{pa}(V)})=\sum_{U\in \mathrm{pa}(V)}\left(\mathds{1}_{y_U}+y_U\right)$ and $\mu_V(\boldsymbol{y}_{\mathrm{pa}(V)})=\sum_{U\in \mathrm{pa}(V)}\left(\mathds{1}_{y_U}-y_U\right)$; finally, for parametrization c) we take $\sigma^2_V=1$ and
\begin{multline*}
\log\frac{p_V}{1-p_V}(\boldsymbol{y}_{\mathrm{pa}(V)})=\sum_{U\in \mathrm{pa}(V)}\left(\mathds{1}_{y_U}+y_U+\frac{y_U^2}{10}\right)+\frac{1}{10}\sum_{\substack{U, W\in\mathrm{pa}(V)\\ U\neq W}}(\mathds{1}_{y_U}+y_U)(\mathds{1}_{y_V}+y_V),\,\,\text{and}
\\
\mu_V(\boldsymbol{y}_{\mathrm{pa}(V)})=\sum_{U\in \mathrm{pa}(V)}\left(\mathds{1}_{y_U}-y_U-\frac{y_U^2}{10}\right)+\frac{1}{10}\sum_{\substack{U, W\in\mathrm{pa}(V)\\ U\neq W}}\left(\mathds{1}_{y_U}\mathds{1}_{y_V}-y_U\mathds{1}_{y_V}-y_V\mathds{1}_{y_U}-y_Vy_U\right).
\end{multline*}
We then normalize the coefficients in the above expressions ($\pm 1,\pm1/10$) so that $\alpha_V$, $\beta_V$, $\log p_V/(1-p_V)$ and $\mu_V$ have means $0$ and $1$, respectively, across the samples. This normalization ensures that the marginal probability of being nonzero, the marginal mean, and the marginal variance for each node are stabilized, in order to show that the DAGs are truly recovered based on the conditional dependency structure instead of additional signals from these marginal quantities. In fact, in the generated samples the marginal probability is about 0.5 and the marginal mean is about 0 for all nodes, and the marginal variance for the nonzero part only is about the same for all except the source node.

In Figure~\ref{plot_identifiability}, we present pairwise scatter plots of one instance of data generated with the chain graph (upper row) and the complete graph (lower row), respectively, both with $(p,\mu,k)$-linear parametrization. Since the true topological ordering is $1\to2\to3\to4\to5$, for clarity we exclude the source and sink nodes ($1$ and $5$) and only include nodes $2$, $3$ and $4$. Plots on the left are plotted in the order $2,3,4$ and those on the right are reversed. In the histograms on the diagonals we only plot the continuous part. 

The scatter plots indicate a slight difference in the respective marginal distributions of nodes 2 and 4 conditioned on node 3 being 0 (and vice versa). This difference intuitively explains how the orientation $2\to 3\to 4$ versus $4\to 3 \to 2$ can be identified. It is worth noting that other than this difference, the marginal statistics for the three nodes are indistinguishable and there is little noticeable difference between plots on the left and on the right.

\begin{figure}
\centering
\includegraphics[width=0.48\textwidth]{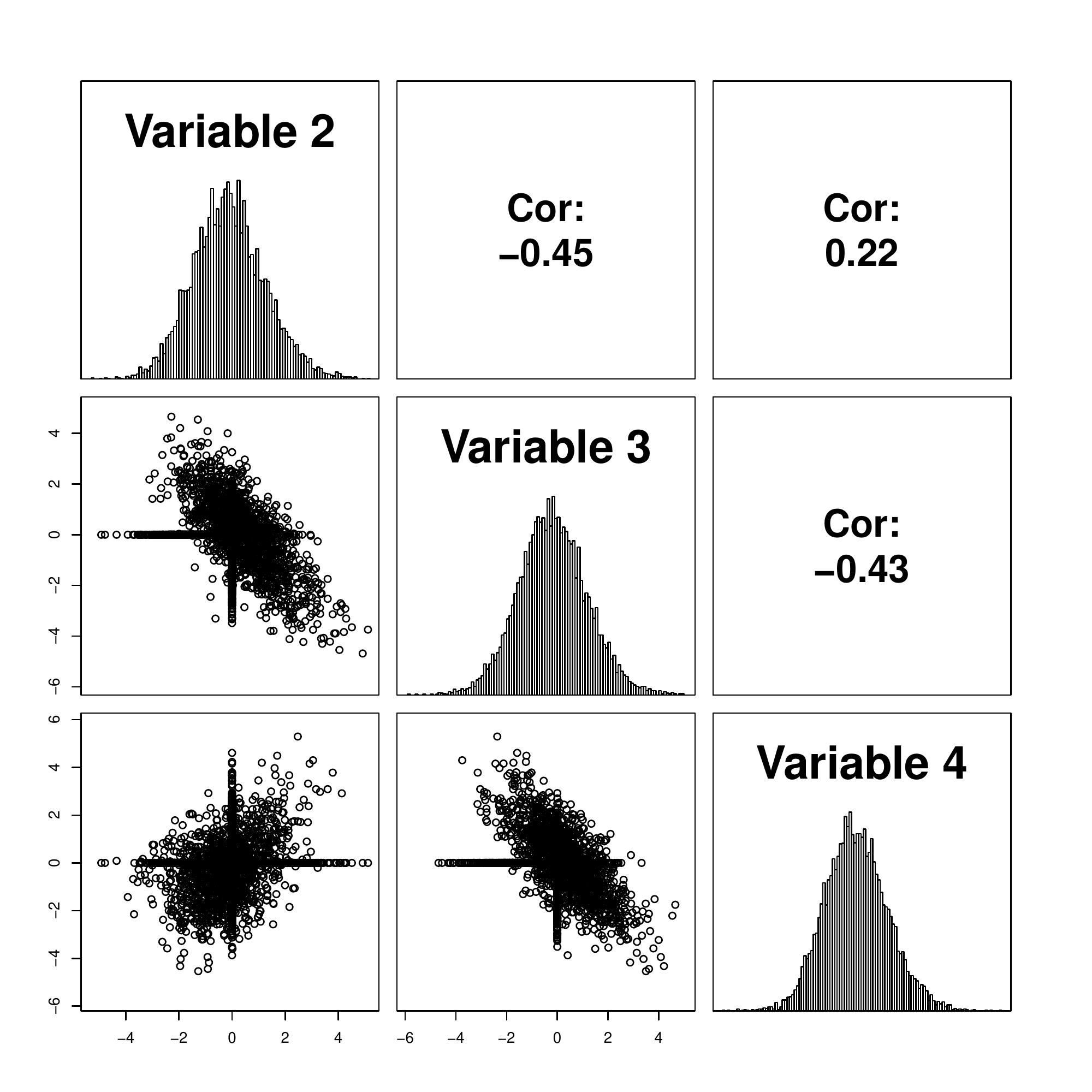}
\includegraphics[width=0.48\textwidth]{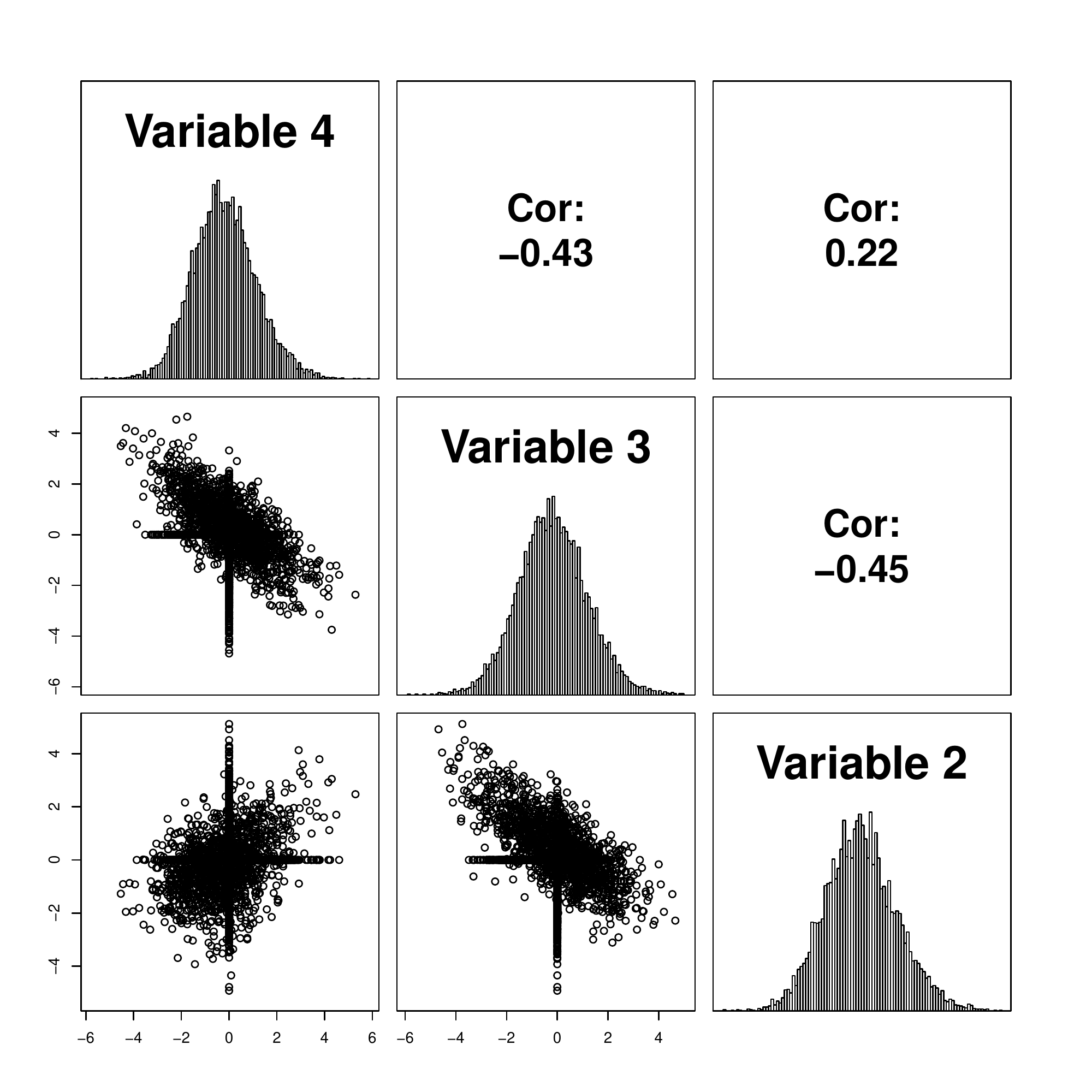}
\vspace{-0.2in}
\includegraphics[width=0.48\textwidth]{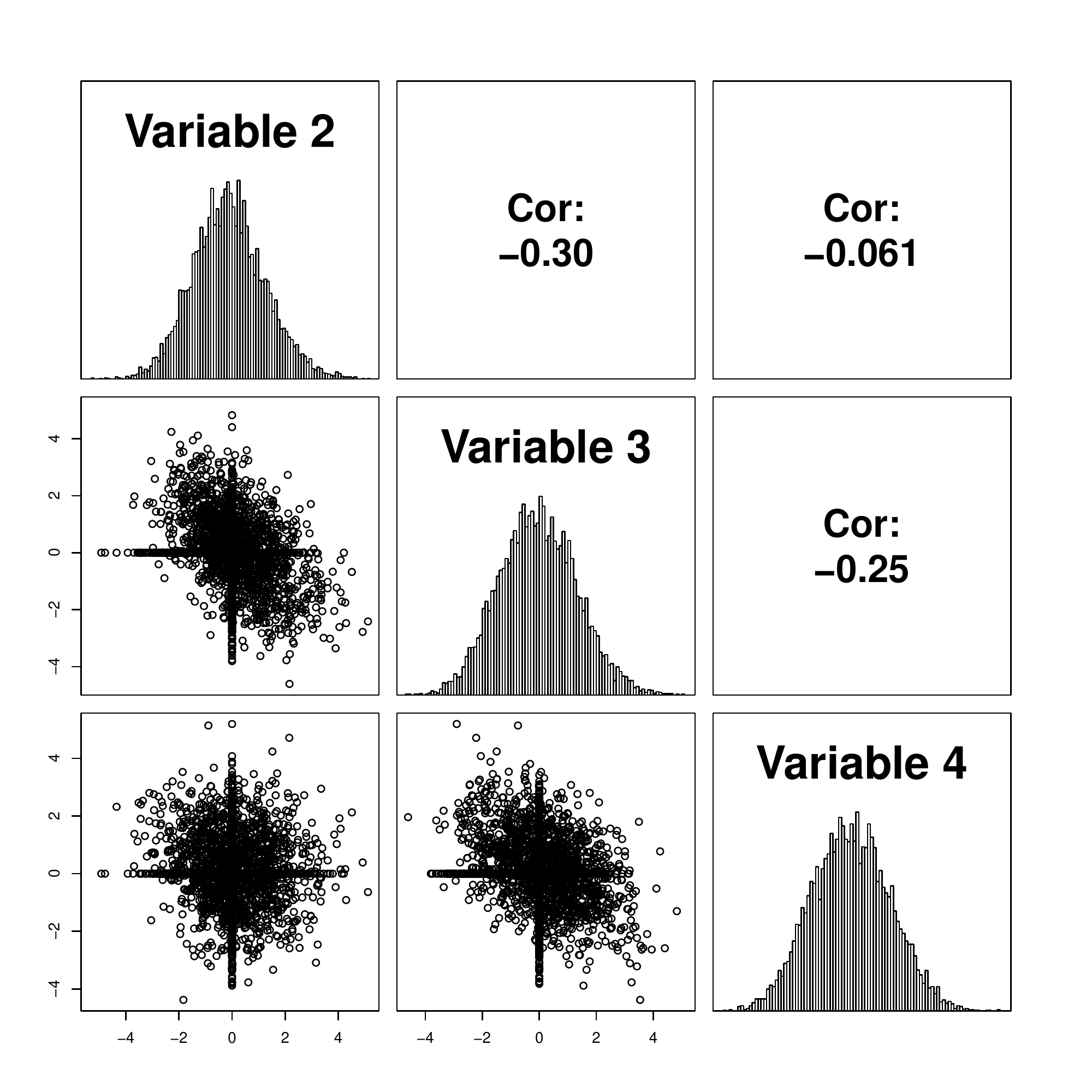}
\includegraphics[width=0.48\textwidth]{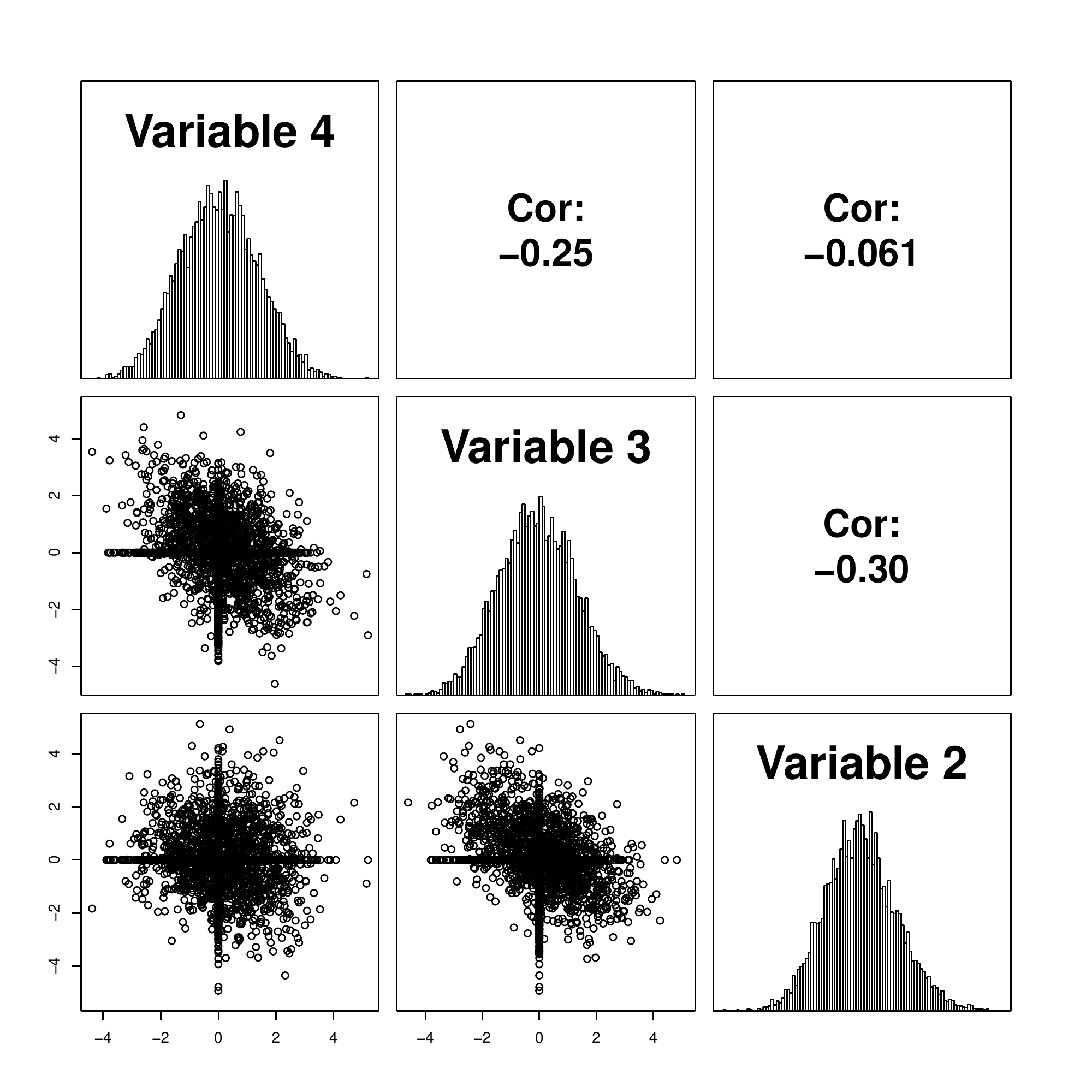}
\caption{Pairwise scatterplots of zero-inflated data generated using chain graphs (upper row) and complete graphs (lower row), both with topological ordering $1\to 2\to 3\to 4\to 5$; only nodes $2$, $3$ and $4$ are plotted. Plots on the left are plotted in the order $2,3,4$, and $4,3,2$ on the right. Only the continuous part is plotted in the histograms on the diagonals. There is little noticeable difference between the histograms and scatter plots when we reverse the graph order, yet our methods can still determine the correct topological ordering.}\label{plot_identifiability}
\end{figure}

\subsection{Estimation}
We use the two graph estimation methods described in Section~\ref{Graph Search}, our self-implemented greedy search (GDS) \citep{chi02} with BIC score, and exhaustive search with dynamic programming \citep{sil06}. 
Details on fitting the hurdle conditionals themselves were presented in Section~\ref{Fitting Hurdle Conditionals}. 

In our simulation, we aim to assess the performance of the different estimation procedures for correctly specific and misspecified parametrizations. To this end, for each combination of true DAG and true data generating parametrization---$(\alpha,\beta,k)$-linear and $(p,\mu,\sigma^2)$-linear and quadratic---we estimate the DAG using all three parametrizations for generating data. For simplicity and given that the simulation results are presented over $B=100$ iterations, stability selection is not used in these experiments. 

\subsection{Results}
Results are shown in Figures~\ref{chain_m5}--\ref{lattice_m9}. 
Each figure has one true underlying DAG from those mentioned in Section \ref{True Underlying DAGs and Distributions}. In all figures, each row indicates one choice of true data generating parametrization---$(\alpha,\beta,k)$-linear, and $(p,\mu,\sigma^2)$-linear and quadratic---and each column shows the results using each estimating parametrization. Thus, plots on the diagonal (with bold titles) correspond to correct parametrizations, where the estimating parametrization agrees with the truth. Off-diagonal plots, in contrast, corresponds to cases where the model parametrization is misspecified.


Since exhaustive search compares all possible DAGs for $m$ nodes, for $n$ large enough it provides an indicator of identifiability. Indeed, the results indicate that in all settings, exhaustive search with correct parametrization almost always identifies the exact DAG for large $n$. In fact, since the $(p,\mu,\sigma^2)$-quadratic parametrization covers the other two, in all cases the graphs can be perfectly recovered using the quadratic estimating parametrization. In contrast, when the underlying truth is quadratic, the graph may not be easily identified from estimates that use the other two parametrizations. This is especially the case for the lattice graph. Comparing the linear parametrizations themselves, $(p,\mu,\sigma^2)$ seems less prone to model misspecification and has the advantage of faster estimation with the help of standard softwares for logistic and linear regressions.

Overall, our simulation studies confirm the identifiability theory (Theorem~\ref{thm_full_identifiability}). In particular, our experiments indicate that exhaustive search performs well. They also indicate that GDS works reasonably well for  sparse graphs but may require larger samples for recovering the structure of complete, or very dense, graphs. While exhaustive search often succeeds with high probability even with small samples, it may not be scalable for large $m$. In such cases, the greedy and faster GDS method, which shows promising results, provides a viable alternative. Utilizing the stability selection method of Section~\ref{Stability Selection} can further improve the GDS results. 


\begin{figure}[H]
  \centering 
   \includegraphics[width=\textwidth]{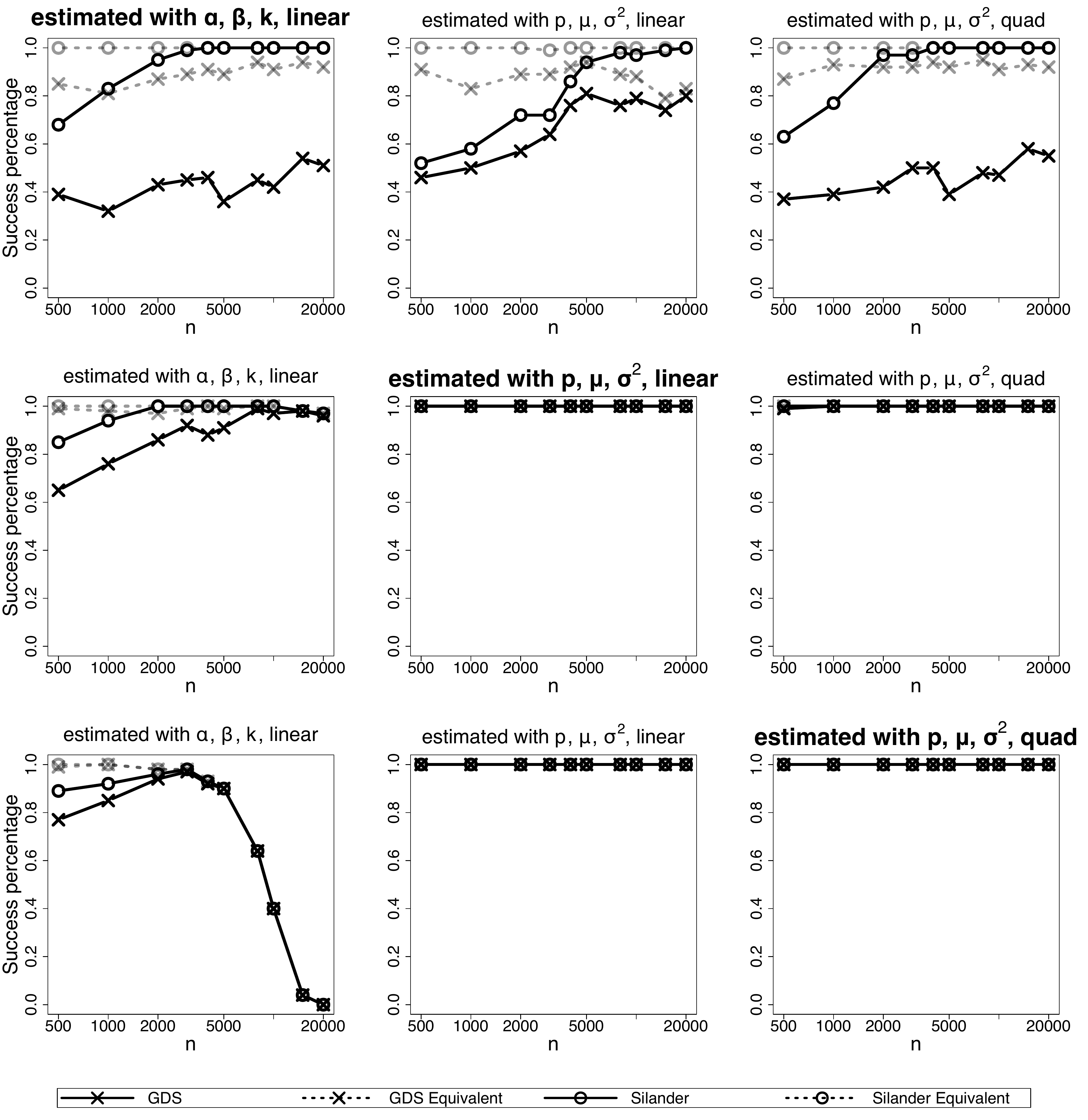}
\vspace{-0.in}
\caption{Chain graph, $m=5$. Each row corresponds to a different generating parametrization, and each column a different estimating parametrization. Generating and estimating parametrizations agree on the diagonal. Solid `$\times$': success rates of exact DAG recovery for greedy search; solid `$\circ$': success rates of exact DAG recovery for exhaustive search; gray dotted lines: success rates for recovery of equivalence class.}\label{chain_m5}
\end{figure}

\begin{figure}[H]
  \centering 
   \includegraphics[width=\textwidth]{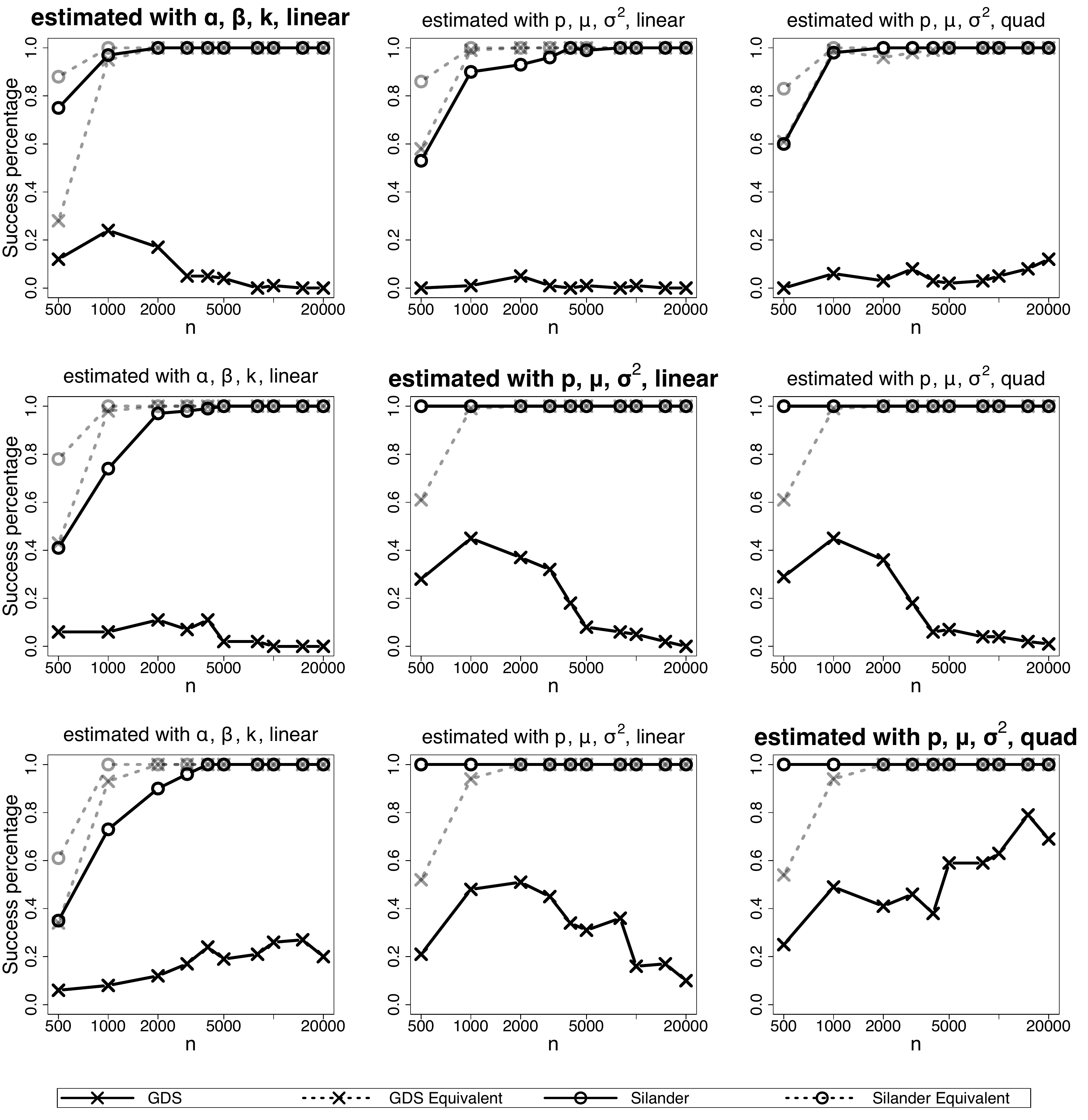}
\vspace{-0.in}
\caption{Complete graph, $m=5$. Each row corresponds to a different generating parametrization, and each column a different estimating parametrization. Generating and estimating parametrizations agree on the diagonal. Solid `$\times$': success rates of exact DAG recovery for greedy search; solid `$\circ$': success rates of exact DAG recovery for exhaustive search; gray dotted lines: success rates for recovery of equivalence class.}\label{complete_m5}
\end{figure}

\begin{figure}[H]
  \centering 
   \includegraphics[width=\textwidth]{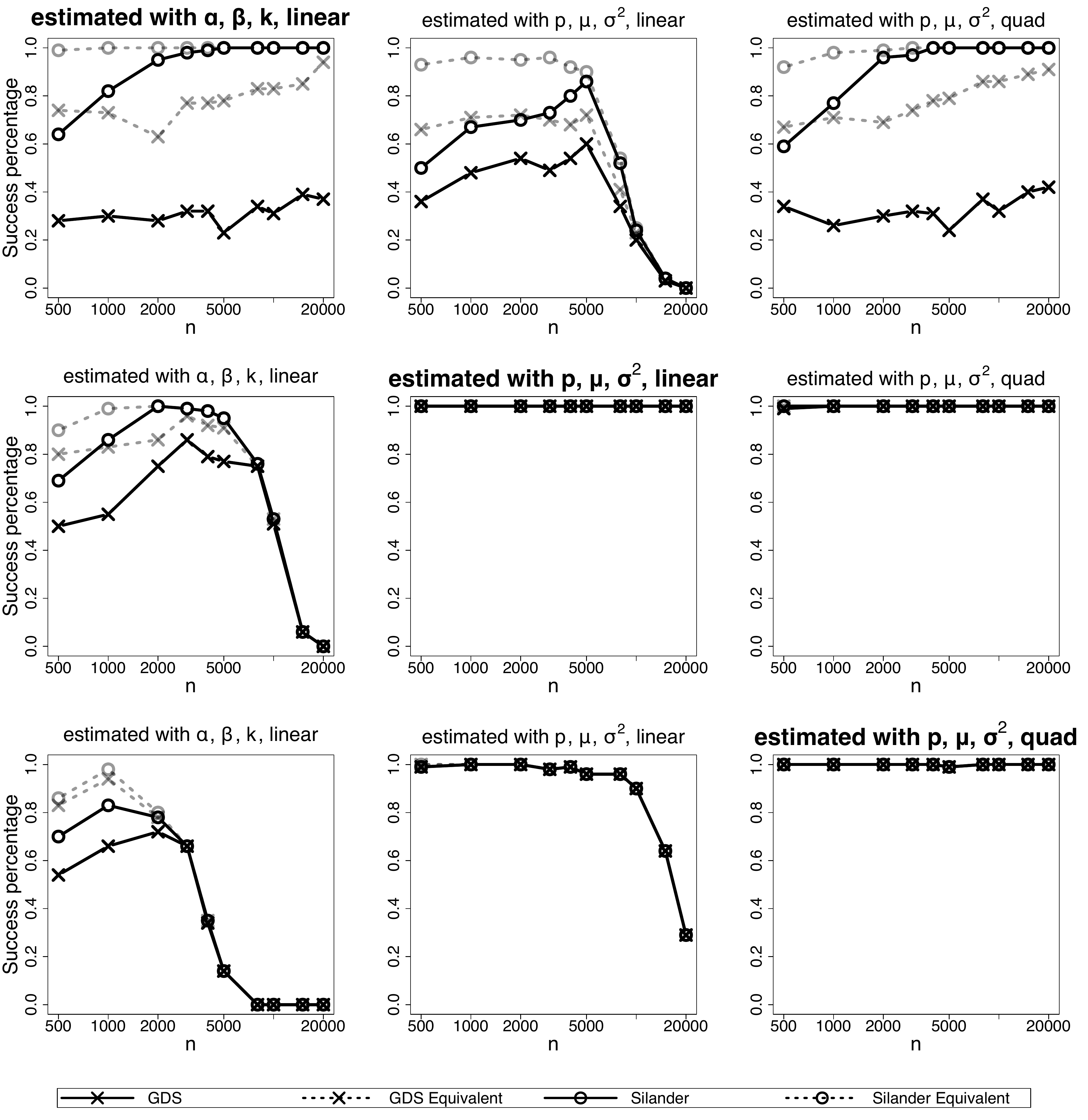}
\vspace{-0.in}
\caption{Lattice graph, $m=9$. Each row corresponds to a different generating parametrization, and each column a different estimating parametrization. Generating and estimating parametrizations agree on the diagonal. Solid `$\times$': success rates of exact DAG recovery for greedy search; solid `$\circ$': success rates of exact DAG recovery for exhaustive search; gray dotted lines: success rates for recovery of equivalence class.}\label{lattice_m9}
\end{figure}

\section{T Helper Cell Data}\label{T Helper Cell Data}
In this section we present the results of applying our model to a T helper cell expression dataset. Specifically, the dataset is considered in \citet{mcd19} and contains both single cell and 10-cell expression measurements for T helper cells for 80 genes in eight healthy donors. We use all 1951 single cell measurements for these donors (a superset of the 465 measurements in \citet{mcd19}) to ensure we have a large enough sample size to produce reliable estimates. In particular, \citet{mcd19} consider only the T-follicular (CXCR5\textsuperscript{+}PD1\textsuperscript{+}) cells that produce high levels of proteins CXCR5 and PD1, while we do not make this restriction. Instead, we add the indicators of CXCR5\textsuperscript{+/-} and PD1\textsuperscript{+/-} as regressors when fitting the conditional distributions. Following \citet{mcd19}, we choose the 61 genes that have at least 5\% zero and 5\% nonzero values.
 
While the measurements are all nonnegative, the minimum, mean, and standard deviation of the nonzero values in the dataset are $7.89$, $18.53$, and $1.91$, respectively. We thus assume zero-inflated conditional Gaussianity without considering the effect of truncation from below at $0$. 

To estimate the DAG structure, we first use the procedure of \citet{mcd19} to identify the connected components in an estimated undirected graph with the same sparsity as the graph therein. We then estimate the directed edges in each connected component using our method. This procedure is justified by the fact that theoretically the connected components for the underlying true undirected and directed graphs coincide. Thus, we generally recommend this strategy in practice, as the connected components can be much more efficiently obtained from the undirected graph. 

We use the $(p,\mu,\sigma^2)$-parametrization as it is more flexible than the $(\alpha,\beta,k)$, and extra fixed covariates and controlling factors can be easily added, since fitting the conditionals only involves linear and logistic regressions. As discussed in Section~\ref{Numerical Experiments}, the $(p,\mu,\sigma^2)$ is also more robust than $(\alpha,\beta,k)$. We use polynomials up to degree three and data-adaptively choose the optimal degree by BIC.

To estimate the DAG, we use the greedy search (GDS) algorithm, which showed promising performance in the simulations of Section~\ref{Numerical Experiments}. 
We also use the stability selection procedure of Section \ref{Stability Selection}, with the goal of controlling the FDR at 10\% for each connected component. For  smaller connected components, if controlling the FDR at 10\% is not possible, we pick the sparsest graph that maximally maintains the connectivity.
Finally, we restrict the node in-degrees to five, in order to both speed up estimation and to constrain the search space. This constraint is motivated by the fact that in gene regulatory networks, each gene is only expected to be regulated by a small number of other genes \citep{alb05}. In contrast, since genetic networks often involve hub genes that regulate many others, we do not restrict the out-degree.


The estimated undirected graph using the procedure of \citet{mcd19} is plotted in the upper half of Figure~\ref{plot_applied}, with edge width and saturation representing the edge strength. In the lower half of the figure, we plot the estimated DAG using our method; the estimate with stability selection and FDR control is shown on the left and the one without stability selection is on the right. Examining the estimates, we find that CD3E is a hub node with degree five in both estimated DAGs, while it has four neighbors in the estimated undirected graph. On the other hand, in the estimate with stability selection, three genes in the largest connected components, namely CD28, JAK1 and STATS5B, are isolated. This is reasonable as they are each associated with only one weak edge in the undirected graph. Moreover, the undirected and directed graph estimates have different thresholds for determining whether an edge is present. For the other nodes, the estimated DAG structures are very similar to the undirected graph estimate. 

\begin{figure}
\centering
\includegraphics[width=1\textwidth]{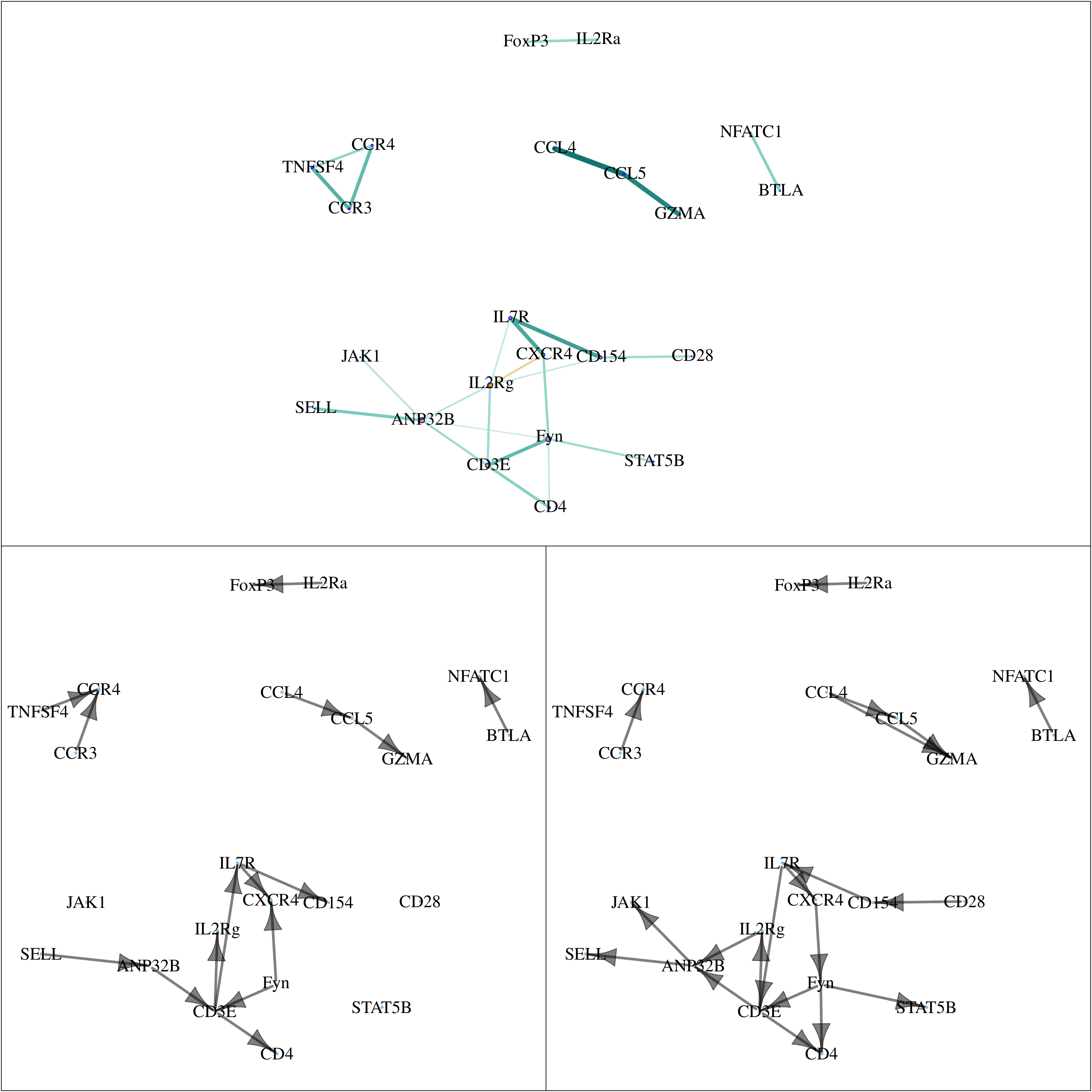}
\caption{Upper: Graph estimated for T helper cell data using undirected zero-inflated graphical models; similar to Figure 6 from \citet{mcd19}
Lower: Directed graph estimated for T helper cell data with (left) and without (right) stability selection and FDR control.}\label{plot_applied}
\vspace{0.2in}
\end{figure}

\section{Discussion}\label{Discussion}
Motivated by the recent advent of single-cell RNA-seq data, in this paper we develop new methods for learning DAGs from zero-inflated data. Our procedures take advantage of two key features of single-cell RNA-seq data, namely, the zero-inflated nature of the data, and the large number of observations from individual samples. 

Our key contribution is establishing identifiability of DAGs from observational zero-inflated data. Specifically, we prove that the exact DAG can be recovered from the joint distribution under reasonable assumptions. We also show that in the most general case, the distributions from which the DAGs are not identifiable only form a small subset, which we prove to be empty in the bivariate and trivariate cases. While our proof uses a very general result on DAGs from \citet{pet14} as its first step, our models do not fit into the framework in that paper; we thus take a different approach that considers the zero-inflation and polynomial structures directly.

Our approach is based on factorizing the joint distribution into zero-inflated conditional Gaussian distributions with parameters polynomial in the parents and their indicators of having nonzero values. We present models in terms of two  parametrizations, one called $(\alpha,\beta,k)$ that is linked to the undirected graphs studied in \citet{mcd19}, and the other called $(p,\mu,\sigma^2)$ that directly models the conditional moments. Both approaches have computational appeal. In particular, the $(\alpha,\beta,k)$-parametrization leads to convex loss functions in the parameters to be estimated, while the $(p,\mu,\sigma^2)$-parametrization offers the additional benefit of allowing one to utilize standard software for logistic and linear regression. We combine these models with two state-of-the-art estimation procedures, namely greedy DAG search (GDS) and exhaustive search with dynamic programming. We also validate our identifiability theory using extensive numerical studies. These experiments indicate that the exhaustive search algorithm is effective in correctly identifying DAGs with small number of nodes. For moderate to large DAGs, the GDS algorithm offers a reasonable alternative, with performance comparable to the exhaustive search when the sample size is large enough. 

Our work opens the door to multiple future research directions and extensions. The first is to prove our conjecture that the sets of distributions from which the DAG is not identifiable are empty also for graphs with more than 3 nodes. The second direction of future research is proving the consistency and investigating finite sample properties of the proposed estimation procedures. Finally, an interesting extension of our model would be to consider zero-inflated distributions under a truncation to the nonnegative orthant $\mathbb{R}_+^m$, which would be of interest for nonnegative \emph{omics} data. The main challenge in this case would be the normalizing constant as a function of the parents in the conditional distributions,  since it would not have a closed-form expression. While this may be resolved by generalizing the \emph{score matching} loss \citep{hyv05,hyv07,lyu09,yu19} 
to data of mixed type, the additional difficulty would lie in proving identifiability and addressing estimation from observational data.

\section*{Acknowledgement}
The authors gratefully acknowledge grant DMS/NIGMS-1561814 from the National Science Foundation (NSF) and grant R01-GM114029 from the National Institutes of Health (NIH).

The authors thank Jonas Peters and Steffen Lauritzen for their input on the theoretical results of this paper, and Andrew McDavid for helpful comments on  the application to T helper cell data.  MD and SY completed part of this work in 
the Department of Mathematical Sciences, University of Copenhagen.


\bibliographystyle{biom} 
\bibliography{Zero_paper}

\begin{thebibliography}{}

\bibitem[\protect\citeauthoryear{Albert}{Albert}{2005}]{alb05}
Albert, R. (2005).
\newblock Scale-free networks in cell biology.
\newblock {\em Journal of Cell Science} {\bf 118,} 4947--4957.

\bibitem[\protect\citeauthoryear{Chen, Drton, and Wang}{Chen
  et~al.}{2019}]{che18}
Chen, W., Drton, M., and Wang, Y.~S. (2019).
\newblock On causal discovery with an equal-variance assumption.
\newblock {\em Biometrika} {\bf 106,} 973--980.

\bibitem[\protect\citeauthoryear{Chickering}{Chickering}{2003}]{chi02}
Chickering, D.~M. (2003).
\newblock Optimal structure identification with greedy search.
\newblock {\em Journal of Machine Learning Research (JMLR)} {\bf 3,} 507--554.

\bibitem[\protect\citeauthoryear{Hyv\"{a}rinen}{Hyv\"{a}rinen}{2005}]{hyv05}
Hyv\"{a}rinen, A. (2005).
\newblock Estimation of non-normalized statistical models by score matching.
\newblock {\em Journal of Machine Learning Research} {\bf 6,} 695--709.

\bibitem[\protect\citeauthoryear{Hyv\"{a}rinen}{Hyv\"{a}rinen}{2007}]{hyv07}
Hyv\"{a}rinen, A. (2007).
\newblock Some extensions of score matching.
\newblock {\em Computational Statistics \& Data Analysis} {\bf 51,} 2499--2512.

\bibitem[\protect\citeauthoryear{Lyu}{Lyu}{2009}]{lyu09}
Lyu, S. (2009).
\newblock Interpretation and generalization of score matching.
\newblock In {\em Proceedings of the Twenty-Fifth Conference on Uncertainty in
  Artificial Intelligence}, pages 359--366. AUAI Press.

\bibitem[\protect\citeauthoryear{Maathuis, Drton, Lauritzen, and
  Wainwright}{Maathuis et~al.}{2019}]{handbook:2019}
Maathuis, M., Drton, M., Lauritzen, S., and Wainwright, M. (2019).
\newblock {\em Handbook of graphical models}.
\newblock Chapman \& Hall/CRC Handbooks of Modern Statistical Methods. CRC
  Press, Boca Raton, FL.

\bibitem[\protect\citeauthoryear{McDavid, Gottardo, Simon, and Drton}{McDavid
  et~al.}{2019}]{mcd19}
McDavid, A., Gottardo, R., Simon, N., and Drton, M. (2019).
\newblock Graphical models for zero-inflated single cell gene expression.
\newblock {\em The Annals of Applied Statistics} {\bf 13,} 848--873.

\bibitem[\protect\citeauthoryear{Okamoto}{Okamoto}{1973}]{oka73}
Okamoto, M. (1973).
\newblock Distinctness of the eigenvalues of a quadratic form in a multivariate
  sample.
\newblock {\em The Annals of Statistics} {\bf 1,} 763--765.

\bibitem[\protect\citeauthoryear{Peters and B{\"u}hlmann}{Peters and
  B{\"u}hlmann}{2013}]{pet13}
Peters, J. and B{\"u}hlmann, P. (2013).
\newblock Identifiability of {G}aussian structural equation models with equal
  error variances.
\newblock {\em Biometrika} {\bf 101,} 219--228.

\bibitem[\protect\citeauthoryear{Peters, Mooij, Janzing, and
  Sch\"{o}lkopf}{Peters et~al.}{2014}]{pet14}
Peters, J., Mooij, J.~M., Janzing, D., and Sch\"{o}lkopf, B. (2014).
\newblock Causal discovery with continuous additive noise models.
\newblock {\em The Journal of Machine Learning Research} {\bf 15,} 2009--2053.

\bibitem[\protect\citeauthoryear{Shah and Samworth}{Shah and
  Samworth}{2013}]{sha13}
Shah, R.~D. and Samworth, R.~J. (2013).
\newblock Variable selection with error control: another look at stability
  selection.
\newblock {\em Journal of the Royal Statistical Society: Series B (Statistical
  Methodology)} {\bf 75,} 55--80.

\bibitem[\protect\citeauthoryear{Shimizu, Hoyer, Hyv\"{a}rinen, and
  Kerminen}{Shimizu et~al.}{2006}]{shi06}
Shimizu, S., Hoyer, P.~O., Hyv\"{a}rinen, A., and Kerminen, A. (2006).
\newblock A linear non-{G}aussian acyclic model for causal discovery.
\newblock {\em Journal of Machine Learning Research} {\bf 7,} 2003--2030.

\bibitem[\protect\citeauthoryear{Silander and Myllym{\"a}ki}{Silander and
  Myllym{\"a}ki}{2006}]{sil06}
Silander, T. and Myllym{\"a}ki, P. (2006).
\newblock A simple approach for finding the globally optimal bayesian network
  structure.
\newblock In {\em Conference on Uncertainty in Artificial Intelligence}, pages
  445--452.

\bibitem[\protect\citeauthoryear{Wang and Drton}{Wang and
  Drton}{2020}]{sam:2020}
Wang, Y.~S. and Drton, M. (2020).
\newblock High-dimensional causal discovery under non-{G}aussianity.
\newblock {\em Biometrika} {\bf 107,} 41--59.

\bibitem[\protect\citeauthoryear{Yu, Drton, and Shojaie}{Yu
  et~al.}{2019}]{yu19}
Yu, S., Drton, M., and Shojaie, A. (2019).
\newblock Generalized score matching for non-negative data.
\newblock {\em Journal of Machine Learning Research} {\bf 20,} 1--70.

\end{thebibliography}

\renewcommand{\thesection}{Appendix}

 \section{Proofs}\label{appendix_proofs}
 In this appendix we present proofs for the theorems and corollaries in the paper.

We first prove the following lemma that states that if two sums of distinct (ignoring the multiplicative constant) exponentials of polynomials in $\boldsymbol{y}\in\mathbb{R}^m$ agree almost everywhere in $\mathbb{R}^m$, then they must have the same number of terms and there must be a 1-1 correspondence between the terms.
\begin{lemma}\label{lem_match}
Let the number of variable be $m\geq 1$ and the degree be $p\geq 1$. Let $\mathcal{D}\equiv\{\boldsymbol{d}\in\mathbb{Z}_{\geq 0}^m:1\leq \sum_{j=1}^m d_j\leq p\}$ be the set of nonnegative integer-valued $m$-vectors with $\ell_1$ norm $\in[1,p]$.
Given a vector $\boldsymbol{a}\in\mathbb{R}^{|\mathcal{D}|}$ indexed by $\boldsymbol{d}\in\mathcal{D}$ (i.e.~$a_{\boldsymbol{d}}\in\mathbb{R}$ for all $\boldsymbol{d}\in\mathcal{D}$), define 
\[f^{(m)}(\boldsymbol{y};\boldsymbol{a})\equiv \exp\left(\sum_{\boldsymbol{d}\in\mathcal{D}}a_{\boldsymbol{d}}\prod_{j=1}y_i^{d_j}\right),\]
the exponential of the corresponding polynomial of degree $\leq p$ in $\boldsymbol{y}\in\mathbb{R}^m$. Note that $f^{(m)}$ does not have a constant term, and has degrees $\boldsymbol{d}\in\mathcal{D}$ and coefficients $\boldsymbol{a}$. 

Suppose we have
\begin{equation}\label{eq_sumf}
\sum_{i=1}^{N_a}a_0^i f^{(m)}(\boldsymbol{y};\boldsymbol{a}^i)=\sum_{i=1}^{N_b}b_0^i f^{(m)}(\boldsymbol{y};\boldsymbol{b}^i)
\end{equation}
for almost every $\boldsymbol{y}\equiv(y_1,\dots,y_m)\in\mathbb{R}^m$ with respect to the Lebesgue measure, where $N_a\geq 0$, $N_b\geq 0$, $\{\boldsymbol{a}^i\}_{i=1}^{N_a}$ are $N_a$ distinct vectors in $\mathbb{R}^{|\mathcal{D}|}$, $\{\boldsymbol{b}^i\}_{i=1}^{N_b}$ are $N_b$ distinct vectors in $\mathbb{R}^{|\mathcal{D}|}$ (otherwise just combine the coefficients), and $a_0^i,b_0^i\in\mathbb{R}\backslash\{0\}$ for all $i$. In other words, both sides of (\ref{eq_sumf}) are a sum of distinct exponentials of polynomials.

Then we must have $N_a=N_b$ and there is a permutation $\pi$ of $\{1,\dots,N_a\}$ such that $\boldsymbol{a}^i=\boldsymbol{b}^{\pi(i)}$ and $a_0^i=b_0^{\pi(i)}$, i.e.~there is a 1-1 correspondence between the summands on both sides of (\ref{eq_sumf}).
\end{lemma}

\begin{proof}[Proof of Lemma \ref{lem_match}] First note that both sides of (\ref{eq_sumf}) are continuous functions, and so is their difference, which is $0$ almost everywhere by assumption. Thus, the inverse image of the open set $\mathbb{R}\backslash\{0\}$ under the difference is also open, and must be the empty set since it has measure 0. (\ref{eq_sumf}) thus holds for all $\boldsymbol{y}\in\mathbb{R}^m$.

We prove by induction on $m$, and first show the result for $m=1$. In this case, $f^{(1)}(y_1;\boldsymbol{a})\equiv\exp(a_1y_1+\dots+a_py_1^p)$, and $\boldsymbol{a}$ is just a $p$-vector. 

First suppose $N_a\neq 0$ and $N_b\neq 0$. Observe that as $x\nearrow+\infty$, if $a_0\neq 0$, the function $a_0\exp(a_1x+\dots+a_px^p)$ goes to
\begin{enumerate}[(i)]
\item $a_0\neq 0$ if $a_1=\dots=a_p=0$, or
\item $0$ if $a_{d_{\max\neq 0}(\boldsymbol{a})}<0$ where ${d_{\max\neq 0}}(\boldsymbol{a})$ is the largest $d\in\{1,\dots,p\}$ such that $a_d\neq 0$, or
\item $+\infty$ if $a_{d_{\max\neq 0}(\boldsymbol{a})}>0$.
\end{enumerate}

Rearrange the terms on the left of (\ref{eq_sumf}) so that for each $1\leq i<j\leq N_a$ we have $(\boldsymbol{a}^i-\boldsymbol{a}^j)_{d_{\max\neq 0}(\boldsymbol{a}^i-\boldsymbol{a}^j)}>0$, and denote this total order as $\boldsymbol{a}^i>\boldsymbol{a}^j$. Rearrange the right-hand side similarly. By the assumption that $\{\boldsymbol{a}^i\}_{i=1}^{N_a}$ are distinct, $\boldsymbol{a}^i-\boldsymbol{a}^j\neq 0$, so $d_{\max\neq 0}(\boldsymbol{a}^i-\boldsymbol{a}^j)$ exists and this rearrangement is possible. Now dividing both sides of (\ref{eq_sumf}) by $f^{(1)}(y_1;\boldsymbol{a}^1)=\exp(a_1^1y_1+\cdots+a_{p}^1y_1^p)$ we have 
\begin{equation}\label{1sf_sf}
a_0^1+\sum_{i=2}^{N_a}a_0^if^{(1)}(y_1;\boldsymbol{a}^i-\boldsymbol{a}^1)=\sum_{i=1}^{N_b}b_0^if^{(1)}(y_1;\boldsymbol{b}^i-\boldsymbol{a}^1).
\end{equation} 

Since $a_0^1\neq 0$, and by the unique maximality of $\boldsymbol{a}^1$, as $y_1\nearrow+\infty$, all terms in the summation on the left go to $0$ (case (ii)). Thus, the right-hand side necessarily also goes to $a_0^1\neq 0$, landing us in case (i) for at least one (and only one because $\boldsymbol{b}^i$ are unique) term on the right, i.e.~$\boldsymbol{b}^i-\boldsymbol{a}^1=\boldsymbol{0}$. (A nonzero finite limit cannot come from a sum of terms that go to $+\infty$ with positive and negative weights, since they must grow at different rates by uniqueness of $\boldsymbol{b}^i-\boldsymbol{a}^1$.) Since summands on both sides are sorted, we must have $\boldsymbol{b}^1=\boldsymbol{a}^1$.

Then (\ref{1sf_sf}) becomes $a_0^1-b_0^1+\sum_{i=2}^{N_a}a_0^i f^{(1)}(y_1;\boldsymbol{a}^i-\boldsymbol{a}^1)=\sum_{i=2}^{N_b}b_0^i f^{(1)}(y_1;\boldsymbol{b}^i-\boldsymbol{a}^1)$. If $a_0^1\neq b_0^1$, by the same reasoning there exists another $i\in\{2,\dots,N_b\}$ such that $\boldsymbol{b}^i-\boldsymbol{a}^1=\boldsymbol{0}$, violating uniqueness of $\{\boldsymbol{b}^i\}_{i=1}^{N_b}$. Thus, $a_0^1=b_0^1$ and $\boldsymbol{a}^1=\boldsymbol{b}^1$, and we have reduced the number of summands on both sides of (\ref{1sf_sf}) by 1 to 
\[\sum_{i=2}^{N_a}a_0^if^{(1)}(y_1;\boldsymbol{a}^i-\boldsymbol{a}^1)=\sum_{i=2}^{N_b}b_0^if^{(1)}(y_1;\boldsymbol{b}^i-\boldsymbol{a}^1).\]

Continuing this process by each time dividing both sides by $f^{(1)}(y_1;\boldsymbol{a}^{j}-\boldsymbol{a}^{j-1})$, we would have matched $\min\{N_a,N_b\}$ pairs of coefficients between the $a$ and the $b$ groups. If $N_a\neq N_b$, assume $N_a>N_b$ without loss of generality, then 
\[\sum_{i=N_b+1}^{N_1}a_0^if^{(1)}(y_1;\boldsymbol{a}^i-\boldsymbol{a}^{N_b})=\mathrm{const}.\]

Here the right-hand side is a constant that could be nonzero, because the argument for $a_0^1=b_0^1$ in our first elimination step does not apply here. Dividing both sides by $f^{(1)}(y_1;\boldsymbol{a}^{N_b+1}-\boldsymbol{a}^{N_b})$, we have $a_0^{N_b+1}+\sum_{i=N_b+2}^{N_1}a_0^if^{(1)}(y_1;\boldsymbol{a}^i-\boldsymbol{a}^{N_b+1})=f^{(1)}(y_1;\boldsymbol{a}^{N_b}-\boldsymbol{a}^{N_b+1})$. By maximality of $\boldsymbol{a}^{N_b+1}$ among $\boldsymbol{a}^{N_b+1},\dots,\boldsymbol{a}^{N_a}$, the left-hand side goes to $a_0^{N_b+1}\neq 0$ as $y_1\nearrow+\infty$, while since $\boldsymbol{a}^{N_b}>\boldsymbol{a}^{N_b+1}$, the right-hand side goes to $+\infty$, a contradiction. Thus, $N_a=N_b$, $a_0^i=b_0^i$ and $\boldsymbol{a}^i=\boldsymbol{b}^i$ for $i=1,\dots,N_a$, proving the $m=1$ case when $N_a\neq 0$ and $N_b\neq 0$.

Now consider the case where one of $N_a$ and $N_b$ is 0; assume without loss of generality that $N_b=0$, then by division by $f^{(1)}(y_1;\boldsymbol{a}^1)$, the right-hand side is constant $0$, while the left-hand side goes to $a_0^1\neq 0$ unless $N_a=0$, so $N_a=N_b=0.$

Now suppose the result holds for some $m-1\geq 1$, and suppose either $N_a\neq 0$ or  $N_b\neq 0$, otherwise there is nothing to prove. We denote $\boldsymbol{a}_1$ as the subvector of $\boldsymbol{a}$ corresponding to $\boldsymbol{d}$ with $d_1\geq 1$, i.e.~$\{a_{\boldsymbol{d}}\}_{\boldsymbol{d}\in\mathcal{D},\,d_1\geq 1}$, and $\boldsymbol{a}_{-1}$ as that of $\boldsymbol{a}$ with $d_1=0$. Separating out the terms involving $y_1$,
\begin{align*}
f^{(m)}(\boldsymbol{y};\boldsymbol{a}^i)&=\exp\left\{
\sum_{d=1}^p\left(\sum_{\boldsymbol{d}\in\mathcal{D},\,d_1=d}a_{\boldsymbol{d}}^i\prod_{j=2}^my_j^{d_j}\right)y_1^{d}\right\}\exp\left(\sum_{\boldsymbol{d}\in\mathcal{D},\,d_1=0}a_{\boldsymbol{d}}^i\prod_{j=2}^my_j^{d_j}\right)\\
=&\,f^{(1)}\left(y_1;\boldsymbol{a}_{1*}^{i}(\boldsymbol{y}_{-1})\right)f^{(m-1)}(\boldsymbol{y}_{-1};\boldsymbol{a}_{-1}^i),
\end{align*}
where $\boldsymbol{a}_{1*}^i(\boldsymbol{y}_{-1}):\mathbb{R}^{m-1}\to\mathbb{R}^p$ is a vector-valued function in $\boldsymbol{y}_{-1}$, with $d$-th coordinate a polynomial $\sum_{\boldsymbol{d}\in\mathcal{D},\,d_1=d}a_{\boldsymbol{d}}^i\prod_{j=2}^my_j^{d_j}$, and coefficients corresponding to $\boldsymbol{a}_1^i$. Note that there is a one-to-one correspondence between such a function $\boldsymbol{a}_{1*}^i$ and vector $\boldsymbol{a}_1^i$.  So we can rewrite (\ref{eq_sumf}) as
\[\sum_{i=1}^{N_a}a_0^if^{(1)}(y_1;\boldsymbol{a}_{1*}^{i}(\boldsymbol{y}_{-1}))f^{(m-1)}(\boldsymbol{y}_{-1};\boldsymbol{a}_{-1}^i)=\sum_{i=1}^{N_b}b_0^if^{(1)}(y_1;\boldsymbol{b}_{1*}^{i}(\boldsymbol{y}_{-1}))f^{(m-1)}(\boldsymbol{y}_{-1};\boldsymbol{b}_{-1}^i)\]
for all $\boldsymbol{y}\in\mathbb{R}^{m}$. Then collecting terms with the same $f^{(1)}$ (same $\boldsymbol{a}_{1}^i$ ($\boldsymbol{a}_{1*}^i$) or $\boldsymbol{b}_{1}^i$ ($\boldsymbol{b}_{1*}^i$)),
\begin{equation}\label{eq_split_ff}
\sum_{\ell=1}^C
f^{(1)}(y_1;\boldsymbol{c}_{1*}^{\ell}(\boldsymbol{y}_{-1}))
\left\{\sum_{j=1}^{n_{\ell}^a}a_0^{k^a_{\ell j}}f^{(m-1)}(\boldsymbol{y}_{-1};\boldsymbol{a}_{-1}^{k^a_{\ell j}})+\sum_{j=1}^{n_{\ell}^b}b_0^{k^b_{\ell j}}f^{(m-1)}(\boldsymbol{y}_{-1};\boldsymbol{b}_{-1}^{k^b_{\ell j}})\right\}=0,
\end{equation}
where $C>0$, each $\boldsymbol{c}_{1}^{\ell}$ (coefficients for $\boldsymbol{c}_{1*}^{\ell}$) is some $\boldsymbol{a}_{1}^{i}$ or $\boldsymbol{b}_{1}^{i}$, and $\{\boldsymbol{c}_{1}^{\ell}\}_{\ell=1}^C$ are distinct. Here, let $\{k_{11}^a,\dots,k_{1,n_{1}^a}^a,\dots,k_{C1}^a,\dots,k_{Cn_{C}^a}^a\}$ be a permutation of $\{1,\dots,N_a\}$, and $\{k_{11}^b,\dots,k_{1,n_{1}^b}^b,\dots,$ \\ $k_{C1}^b,\dots,k_{Cn_{C}^b}^b\}$ a permutation of $\{1,\dots,N_b\}$. 

Since $\{\boldsymbol{c}_{1}^{\ell}\}_{\ell=1}^C$ are distinct, $\{\boldsymbol{c}_{1*}^{\ell}\}_{\ell=1}^C$ are distinct finite polynomials in $\boldsymbol{y}_{-1}\in\mathbb{R}^{m-1}$. For each pair of such distinct polynomials, the lemma of \citet{oka73} implies that they only agree at a Lebesgue-null subset of $\mathbb{R}^{n-1}$, so all polynomials are distinct except on a null set. Thus, for almost every fixed $\boldsymbol{y}_{-1}\in\mathbb{R}^{m-1}$, the left-hand side of (\ref{eq_split_ff}) is a sum of $C>0$ distinct $f^{(1)}$'s in $y_1$ multiplied by constant weights depending on $\boldsymbol{y}_{-1}$. But the right-hand side is a sum of $0$ terms, so by the result for $m=1$ we necessarily have 
\begin{equation}\label{eq_split_2:m}
\sum_{j=1}^{n_{\ell}^a}a_0^{k^a_{\ell j}}f^{(m-1)}(\boldsymbol{y}_{-1};\boldsymbol{a}_{-1}^{k^a_{\ell j}})=\sum_{j=1}^{n_{\ell}^b}-b_0^{k^b_{\ell j}}f^{(m-1)}(\boldsymbol{y}_{-1};\boldsymbol{b}_{-1}^{k^b_{\ell j}})
\end{equation} for all $\ell=1,\dots,C$ for almost every $\boldsymbol{y}_{-1}$. Fixing $\ell\in\{1,\dots,C\}$, for any $1\leq j_1<j_2\leq n_\ell^a$, $\boldsymbol{a}^{k_{\ell j_1}^a}\neq \boldsymbol{a}^{k_{\ell j_2}^a}$  and $\boldsymbol{a}_1^{k_{\ell j_1}^a}=\boldsymbol{a}_1^{k_{\ell j_2}^a}$ implies $\boldsymbol{a}_{-1}^{k_{\ell j_1}^a}\neq \boldsymbol{a}_{-1}^{k_{\ell j_2}^a}$, and similarly for $\boldsymbol{b}$. Thus, each term on the left-hand side of (\ref{eq_split_2:m}) has its unique coefficients, and similarly for the right-hand side. Since (\ref{eq_split_2:m}) holds for almost every $\boldsymbol{y}_{-1}$, by the result for $m-1$ variables, we must have $n_\ell^a=n_\ell^b$ and each $a_0^{k_{\ell j}^a}=b_0^{k_{\ell\pi(j)}^b}$ and $\boldsymbol{a}_{-1}^{k_{\ell j}^a}=\boldsymbol{b}_{-1}^{k_{\ell\pi(j)}^b}$ for some permutation $\pi$ of $\{1,\dots,n_\ell^a\}$, which in turn implies $\boldsymbol{a}^{k_{\ell j}^a}=\boldsymbol{b}^{k_{\ell\pi(j)}^b}$ for all $j=1,\dots,n_{\ell}^a$ by construction of the groups $\ell=1,\cdots,C$. Since this holds for all $\ell$, $N_a=\sum_{\ell=1}^Cn_\ell^a=\sum_{\ell=1}^Cn_\ell^b=N_b$, and we have thus again matched each $\boldsymbol{a}^\ell$ with a $\boldsymbol{b}^{\ell}$ as well as the corresponding $a_0$'s with $b_0$'s. This ends the proof for $m$, and the entire proof.
\end{proof}

\begin{proof}[Proof of Theorem \ref{thm_full_identifiability}]
Suppose $\mathcal{G}$ and $\mathcal{G}'$ have the same node set $\mathcal{V}$ and are Markov equivalent, otherwise the distributions represented by them are trivially not identical.

Now suppose $p(\boldsymbol{Y})$ is Markov and faithful with respect to $\mathcal{G}$ and $\mathcal{G}'$, and factorize w.r.t.~both graphs with \emph{strong Hurdle polynomial} parameters. Then by Proposition \ref{prop_peters}, there exist $V_1$ and $V_2$ such that $V_1\to  V_2$ in $\mathcal{G}$, $V_2\to V_1$ in $\mathcal{G}'$ and $\mathcal{P}\equiv\mathrm{pa}_{\mathcal{G}}(V_2)\backslash\{V_1\}=\mathrm{pa}_{\mathcal{G}'}(V_1)\backslash\{V_2\}$. Following the arguments in the proof of Proposition \ref{prop_peters} in \citet{pet14}, recursively marginalizing out nodes without children but having the same parents in both graphs, we eventually obtain structures as follows, where $A$ and $B$ are some unknown node sets and $V_2$ does not have any children in Graph $\mathcal{G}$.

\begin{alignat*}{3}
&\mathbf{V}_{A}\hspace{0.13in} \mathbf{V}_\mathcal{P} \quad \quad && \mathbf{V}_A\hspace{0.13in} \mathbf{V}_\mathcal{P} \\
&\-\hspace{0.05in}\searrow  \diagup \searrow \quad\quad && \-\hspace{0.05in}\nwarrow  \,\,\swarrow \diagdown  \\
&\-\hskip0.2in V_1 \to  V_2 \quad \quad && \-\hskip0.2in V_1 \leftarrow  V_2  \\
&\-\hskip0.2in \downarrow\quad \quad && \-\hskip0.2in \downarrow \\
&\-\hskip0.2in \mathbf{V}_B  \quad \quad &&\-\hskip0.2in \mathbf{V}_B\\
&\-\hskip0.06in \text{Graph }\mathcal{G}  \quad \quad &&\-\hskip0.06in \text{Graph }\mathcal{G}'
\end{alignat*}

We consider the $(\alpha,\beta,k)$-parametrization only, since the result for the $(p,\mu,\sigma^2)$ naturally follows from their relationship (\ref{eq_abk_pms_relationship}). For notational simplicity write $V_1$ and $V_2$ as nodes 1 and 2. Suppose after marginalization above we are left with nodes $\mathcal{V}_0\subseteq\mathcal{V}$ which include $1$, $2$, $\mathbf{V}_A$, $\mathbf{V}_B$ and $\mathbf{V}_{\mathcal{P}}$ illustrated above. Now let $Y_U=0$ for all $U\in\mathcal{V}_0\backslash\{2\}$, and let $Y_2\neq 0$.
Then the joint distribution $p(Y_2=y_2\neq 0,\boldsymbol{y}_{\mathcal{V}_0}=\boldsymbol{0})$ using $\mathcal{G}$ is proportional to
\begin{align*}
&\,\left.\prod_{V\in\mathcal{V}_0}\frac{\exp\{\alpha_V(\boldsymbol{y}_{\mathrm{pa}_{\mathcal{G}}(V)})\mathds{1}_{y_v}+\beta_v(\boldsymbol{y}_{\mathrm{pa}_{\mathcal{G}}(V)})y_V-k_Vy_V^2/2\}}{\sqrt{2\pi/k_V}\exp\{\alpha_V(\boldsymbol{y}_{\mathrm{pa}_{\mathcal{G}}(V)})+\beta_V(\boldsymbol{y}_{\mathrm{pa}_{\mathcal{G}}(V)})^2/(2k_V)\}+1}\right|_{y_2\neq 0,\boldsymbol{y}_{\mathcal{V}_0\backslash\{2\}}=\boldsymbol{0}}\\
\propto&\,\exp\{\beta_2(\boldsymbol{0})y_2-k_2y_2^2/2\}
\end{align*}
since 2 does not have any child in $\mathcal{G}$. But using $\mathcal{G}'$, the same joint distribution is proportional to
\begin{multline*}
\left.\prod_{V\in\mathcal{V}_0}\frac{\exp\{\alpha'_V(\boldsymbol{y}_{\mathrm{pa}_{\mathcal{G}'}(V)})\mathds{1}_{y_V}+\beta'_V(\boldsymbol{y}_{\mathrm{pa}_{\mathcal{G}'}(V)})y_V-k'_Vy_V^2/2\}}{\sqrt{2\pi/k'_V}\exp\{\alpha'_V(\boldsymbol{y}_{\mathrm{pa}_{\mathcal{G}'}(V)})+\beta'_V(\boldsymbol{y}_{\mathrm{pa}_{\mathcal{G}'}(V)})^2/(2k'_V)\}+1}\right|_{y_2\neq 0,\boldsymbol{y}_{\mathcal{V}_0\backslash\{2\}}=\boldsymbol{0}}\\
\propto\exp\{\beta'_2(\boldsymbol{0})y_2-k'_2y_2^2/2\}\\
\times \prod_{U\in \mathcal{P}\cup\{1\}, \,2\in\mathrm{pa}_{\mathcal{G}'}(U)}\frac{1}{\sqrt{2\pi/k'_U}\exp\{\alpha'_U(y_2,\boldsymbol{0})+\beta'_U(y_2,\boldsymbol{0})^2/(2k'_U)\}+1}\
\end{multline*}
where in the case where $\mathrm{pa}_{\mathcal{G}'}(2)=\varnothing$ replace $\alpha'_2(\boldsymbol{0})$ and $\beta'_2(\boldsymbol{0})$ by constants $\alpha'_2$ and $\beta'_2$, and $\alpha_U'(y_2,\boldsymbol{0})$ and $\beta_U'(y_2,\boldsymbol{0})$ denote setting all parents other than 2 in the Hurdle polynomials $\alpha_U'$ and $\beta_U'$ to $\boldsymbol{0}$. Since the two joint distributions derived from both graphs must be proportional to each other, we get for $y_2\neq0$
\begin{multline}\label{eq_cannot_match}
\exp\left[y_2\{\beta'_2(\boldsymbol{0})-\beta_2(\boldsymbol{0})\}-(k'_2-k_2)y_2^2/2\right]\\
\propto\prod_{U\in \mathcal{P}\cup\{1\}, \,2\in\mathrm{pa}_{\mathcal{G}'}(U)}\left[\sqrt{2\pi/k'_U}\exp\left\{\alpha'_U(y_2,\boldsymbol{0})+\beta'_U(y_2,\boldsymbol{0})^2/(2k'_U)\right\}+1\right].
\end{multline}
Note that $2\in\mathrm{pa}_{\mathcal{G}'}(1)$ and thus the product on the right of (\ref{eq_cannot_match}) has at least one term. Thus, supposing that for at least one of $U\in\mathcal{P}\cup\{1\}$ such that $2\in\mathrm{pa}_{\mathcal{G}'}(U)$, $\alpha'_U(Y_2,\boldsymbol{0})+\beta'_U(Y_2,\boldsymbol{0})^2/(2k_{U}')$ is nonconstant in $Y_2\neq 0$, then the right-hand side of (\ref{eq_cannot_match}) can be expanded into a sum of at least two exponentials of polynomials in $y_2$ (including the constant 1 as a degenerated exponential polynomial), while the left-hand side is a single polynomial in $y_2$. This is a contradiction according to Lemma \ref{lem_match}, and thus the assumption of having \emph{strong Hurdle polynomials} as the parameters in the Hurdle conditionals implies that $p(\boldsymbol{Y})$ cannot be represented by both $\mathcal{G}$ and $\mathcal{G}'$, which ends the proof.

\end{proof}

\begin{proof}[Proof of Theorem \ref{id_full}]

As in the proof of Theorem \ref{thm_full_identifiability} using Proposition \ref{prop_peters}, under the assumptions there exist $V_1$ and $V_2$ such that $\mathcal{P}\equiv\mathrm{pa}_{\mathcal{G}}(V_2)\backslash\{V_1\}=\mathrm{pa}_{\mathcal{G}'}(V_1)\backslash\{V_2\}$ with $V_1\to V_2$ in $\mathcal{G}$ and $V_2\to V_1$ in $\mathcal{G}'$. Following the arguments in the proof of Proposition \ref{prop_peters} in \citet{pet14}, recursively marginalizing out nodes without children but having the same parents in both graphs, we again obtain structures as follows.

\begin{alignat*}{3}
&\mathbf{V}_{A}\hspace{0.13in} \mathbf{V}_\mathcal{P} \quad \quad && \mathbf{V}_A\hspace{0.13in} \mathbf{V}_\mathcal{P} \\
&\-\hspace{0.05in}\searrow  \diagup \searrow \quad\quad && \-\hspace{0.05in}\nwarrow  \,\,\swarrow \diagdown  \\
&\-\hskip0.2in V_1 \to  V_2 \quad \quad && \-\hskip0.2in V_1 \leftarrow  V_2  \\
&\-\hskip0.2in \downarrow\quad \quad && \-\hskip0.2in \downarrow \\
&\-\hskip0.2in \mathbf{V}_B  \quad \quad &&\-\hskip0.2in \mathbf{V}_B\\
&\-\hskip0.06in \text{Graph }\mathcal{G}  \quad \quad &&\-\hskip0.06in \text{Graph }\mathcal{G}'
\end{alignat*}


To ease the notation assume we again write $V_1=1$ and $V_2=2$.
Note that the distribution of each node conditional on some other nodes is the sum of a point mass at $0$ and a continuous distribution over $\mathbb{R}$, which follows by induction and the fact that the indefinite integral of a continuous density is continuous and that the sum of continuous densities is continuous. We focus on the continuous components, and wish to reach the conclusion using the factorization 
\begin{align*}
P(y_{1},y_{2}|\boldsymbol{Y}_\mathcal{P}=\boldsymbol{y}_\mathcal{P})&=P(y_{1}|\boldsymbol{Y}_\mathcal{P}=\boldsymbol{y}_P)P(y_2|y_1,\boldsymbol{Y}_\mathcal{P}=\boldsymbol{y}_\mathcal{P})\\
&=P(y_2|\boldsymbol{Y}_\mathcal{P}=\boldsymbol{y}_\mathcal{P})P(y_1|y_2,\boldsymbol{Y}_\mathcal{P}=\boldsymbol{y}_\mathcal{P}),
\end{align*}
where the second terms in both decompositions are a regular Hurdle conditional w.r.t.~$\mathcal{G}$  and $\mathcal{G}'$, respectively, and we write the first terms as 
\[P(y_1|\boldsymbol{Y}_\mathcal{P}=\boldsymbol{y}_\mathcal{P})\propto\exp\{\mathds{1}_{y_1}\delta_1+f_1(y_1)\}\]
and
\[P(y_2|\boldsymbol{Y}_\mathcal{P}=\boldsymbol{y}_\mathcal{P})\propto\exp\{\mathds{1}_{y_2}\delta_2'+f_2'(y_1)\}
\]
in terms of the conditional densities w.r.t.~$\lambda$. Here $f_1$ and $f_2'$ are continuous functions in $\mathbb{R}$ with no additive constant term, and $\delta_1$ and $\delta_2'$ are constants. 

We prove the results in the $(\alpha,\beta,k)$-parameterization only, since results for the $(p,\mu,\sigma^2)$-parameterization would follow from their relationship (\ref{eq_abk_pms_relationship}). In our model, we assumed the $\alpha$ and $\beta$ parameters for each node to be polynomial in the parents and their indicators. We also assumed that for each node, either the $\beta$ function is nonconstant in any of the parents, or $\alpha$ depends on the value of all of its parents. 

Consider a generic $\beta$ function associated with some generic parent set $\mathcal{P}\equiv \mathcal{P}_1\sqcup \{p_0\}$ with $p_0\not\in\mathcal{P}_1\neq\varnothing$ and suppose that $\beta$ is nonconstant in any of $\mathcal{P}$, and write $\beta(\boldsymbol{y}_{\mathcal{P}})$ equivalently as $\beta(y_{p_0},\boldsymbol{y}_{\mathcal{P}_1})$. Then $\beta(\boldsymbol{y}_{\mathcal{P}})$ has the form $\beta_{-1}(\boldsymbol{y}_{\mathcal{P}_1})+\beta_{0}(\boldsymbol{y}_{\mathcal{P}_1})\mathds{1}_{y_1}+\sum_{i=1}^k\beta_{i}(\boldsymbol{y}_{\mathcal{P}_1})y_1^i$, where by construction $\beta_{-1}$ through $\beta_{k}$ are (potentially constant or even zero) Hurdle polynomials in $\boldsymbol{y}_{\mathcal{P}_1}$, but there must exist some $j=0,\dots,k$ such that $\beta_{j}$ is nonzero. By the lemma of \citet{oka73}, $\beta_{j}(\boldsymbol{y}_{\mathcal{P}_1})\neq 0$ for (Lebesgue) almost every $\boldsymbol{y}_{\mathcal{P}_1}\in\mathbb{R}^{|\mathcal{P}_1|}$. Thus, $\beta(y_{p_0},\boldsymbol{y}_{\mathcal{P}_1})$ is nonconstant in $y_{p_0}$ for almost every $\boldsymbol{y}_{\mathcal{P}_1}\in\mathbb{R}^{|P_1|}$. Formally, define 
\[\mathcal{Y}_{\beta,p_0,\mathcal{P}_1}\equiv\left\{\boldsymbol{y}_{\mathcal{P}_1}\in\mathbb{R}^{|\mathcal{P}_1|}:\beta\left(y_{p_0},\boldsymbol{y}_{\mathcal{P}_1}\right)\text{ nonconstant function in }y_{p_0}\right\}.\] Thus $\mathbb{R}^{|\mathcal{P}_1|}\backslash\mathcal{Y}_{\beta,p_0,\mathcal{P}_1}$ has zero Lebesgue measure assuming $\beta$ is nonconstant in its any of $\mathcal{P}$. Hence, by a similar argument, under the assumptions of the theorem, letting
\begin{multline*}
\mathcal{Y}_{\alpha,\beta,p_0,\mathcal{P}_1}\equiv\left\{\boldsymbol{y}_{\mathcal{P}_1}\in\mathbb{R}^{|\mathcal{P}_1|}:\beta\left(y_{p_0},\boldsymbol{y}_{\mathcal{P}_1}\right)\text{ nonconstant function in }y_{p_0}\text{ or }\right.\\
\left.\alpha\left(y_{p_0},\boldsymbol{y}_{\mathcal{P}_1}\right)\text{ depends on the value of }y_{p_0}\right\},
\end{multline*}
$\mathbb{R}^{|\mathcal{P}_1|}\backslash\mathcal{Y}_{\alpha,\beta,p_0,\mathcal{P}_1}$ has zero Lebesgue measure.

Now we go back to $\mathcal{G}$ and $\mathcal{G}'$. Suppose $\mathcal{P}\neq\varnothing$ and that the Hurdle density of node $2$ conditional on $\{1\}\sqcup \mathcal{P}$ in $\mathcal{G}$ have $\alpha$ and $\beta$ parameters $\alpha_2(y_1,\boldsymbol{y}_\mathcal{P})$ and $\beta_2(y_1,\boldsymbol{y}_\mathcal{P})$, and let those for $1$ conditional on $\{2\}\sqcup \mathcal{P}$ in $\mathcal{G}'$ be $\alpha_1'(y_2,\boldsymbol{y}_\mathcal{P})$ and $\beta_1'(y_2,\boldsymbol{y}_\mathcal{P})$. We also denote $\mathcal{Y}_*\equiv\mathcal{Y}_{\alpha_2,\beta_2,1,\mathcal{P}}\cap\mathcal{Y}_{\alpha_1',\beta_1',2,\mathcal{P}}$, which by discussion above contains almost every $\boldsymbol{y}_\mathcal{P}\subset\mathbb{R}^{|\mathcal{P}|}$. 

From now on we thus fix $\boldsymbol{y}_\mathcal{P}\in\mathcal{Y}_*$ and condition on $\boldsymbol{Y}_\mathcal{P}=\boldsymbol{y}_\mathcal{P}$, and omit the dependency of the $\alpha$ and $\beta$ functions on $\mathcal{P}$, and write them as scalar functions instead notation-wise. By discussion above, $\beta_2$  becomes a nonconstant function in $y_1$ and $\beta_1'$ becomes a nonconstant function in $y_2$. Note that for $\mathcal{P}=\varnothing$, we do not fix or condition on any parent variables and $\alpha_1'$, $\alpha_2$, $\beta_1'$ and $\beta_2$ are automatically univariate functions, with $\beta_1'$ and $\beta_2$ nonconstant by assumption.

The joint density of $P(y_1,y_2|\boldsymbol{Y}_\mathcal{P}=\boldsymbol{y}_\mathcal{P})$ w.r.t.~$\lambda$ thus has two characterizations (up to normalizing constants)

\begin{multline}\label{eq_joint_dist_equal_general}
\frac{\exp\{\mathds{1}_{y_1}\delta_{1}+f_1(y_1)+\mathds{1}_{y_2}\alpha_2(y_1)+y_2\beta_2(y_1)-y_2^2k_2/2\}}{\sqrt{2\pi/k_2}\exp\{\alpha_2(y_1)+\beta_2(y_1)^2/(2k_2)\}+1}\\
\propto\frac{\exp\{\mathds{1}_{y_2}\delta_2'+f_2'(y_2)+\mathds{1}_{y_1}\alpha_1'(y_2)+y_1\beta_1'(y_2)-y_1^2k_1'/2\}}{\sqrt{2\pi/k_1'}\exp[\alpha_1'(y_2)+\{\beta_1'(y_2)\}^2/(2k_1')]+1},
\end{multline}
where $\alpha_2(y_1)$ has the form $c_{\alpha_2,-1}+c_{\alpha_2,0}\mathds{1}_{y_1}+c_{\alpha_2,1} y_1+\dots+c_{\alpha_2,k}  y_1^k$ with coefficients being polynomials in $\boldsymbol{y}_\mathcal{P}$ and their indicators (or constants if $\mathcal{P}=\varnothing$), and similarly for $\beta_2(y_1)$, $\alpha_1'(y_2)$ and $\beta_1'(y_2)$. Note that if the values of $\mathds{1}_{y_1}$ and $\mathds{1}_{y_2}$ are given, these four functions are just polynomials in $y_1$ and $y_2$, respectively. 

First condition on the event $\mathds{1}_{y_1}=\mathds{1}_{y_2}=1$ that has a positive probability. Then (\ref{eq_joint_dist_equal_general}) becomes
\begin{multline}\label{eq_joint_dist_equal_w1_general}
\frac{\exp\{f_1(y_1)+\alpha_2(y_1)+y_2\beta_2(y_1)-y_2^2k_2/2\}}{\sqrt{2\pi/k_2}\exp\{\alpha_2(y_1)+\beta_2(y_1)^2/(2k_2)\}+1}\mathds{1}_{y_1}\mathds{1}_{y_2},\\
\propto\frac{\exp\{f_2'(y_2)+\alpha_1'(y_2)+y_1\beta_1'(y_2)-y_1^2k_1'/2\}}{\sqrt{2\pi/k_1'}\exp\{\alpha_1'(y_2)+(\beta_1'(y_2))^2/(2k_1')\}+1}\mathds{1}_{y_1}\mathds{1}_{y_2},
\end{multline}
for all $(y_1,y_2)\in(\mathbb{R}\backslash\{0\})^2$. (\ref{eq_joint_dist_equal_w1_general}) has the form 
\[\frac{\exp\{f_1(y_1)+P_1(y_1,y_2)\}}{\exp\{P_2(y_1)\}+1}=\frac{\exp\{f_2'(y_2)+P_3(y_1,y_2)\}}{\exp\{P_4(y_2)\}+1},\]
where $P_1$ and $P_3$ are polynomials in $y_1$ and $y_2$ simultaneously, possibly with interactions from the $y_2\beta_2(y_1)$ and $y_1\beta'_1(y_2)$ terms, and $P_2$ and $P_4$ are univariate polynomials in $y_1$, $y_2$, respectively. By cross-multiplication, 
\begin{multline}\label{eq_cross_mult_general}
\exp\{f_1(y_1)+P_1(y_1,y_2)+P_4(y_2)\}+\exp\{f_1(y_1)+P_1(y_1,y_2)\}\\
=\exp\{f_2'(y_2)+P_3(y_1,y_2)+P_2(y_1)\}+\exp\{f_2'(y_2)+P_3(y_1,y_2)\}.
\end{multline}
Differentiating both sides of (\ref{eq_cross_mult_general}) with respect to $y_1$,
\begin{align}\label{eq_cross_mult_diff_general}
&\,\left[\frac{\partial }{\partial y_1}\left\{f_1(y_1)+P_1(y_1,y_2)\right\}\right]\exp\left\{f_1(y_1)+P_1(y_1,y_2)+P_4(y_2)\right\}\nonumber\\
&\pushright{+\exp\{f_1(y_1)+P_1(y_1,y_2))\}\nonumber}\\
=&\,\left[\frac{\partial}{\partial y_1}\left\{P_3(y_1,y_2)+P_2(y_1)\right\}\right]\exp\left\{f_2'(y_2)+P_3(y_1,y_2)+P_2(y_1)\right\}\nonumber\\
&\pushright{+\left\{\frac{\partial}{\partial y_1}P_3(y_1,y_2)\right\}\exp\left\{f_2'(y_2)+P_3(y_1,y_2)\right\}.\quad\quad}
\end{align}
Plugging (\ref{eq_cross_mult_general}) into the left-hand side of (\ref{eq_cross_mult_diff_general}),
\begin{align*}
&\,\left[\frac{\partial }{\partial y_1}\left\{f_1(y_1)+P_1(y_1,y_2)\right\}\right]\left[\exp\left\{f_2'(y_2)+P_3(y_1,y_2)+P_2(y_1)\right\}\right.\nonumber\\
&\pushright{+\left.\exp\left\{f_2'(y_2)+P_3(y_1,y_2)\right\}\right]\nonumber}\\
=&\,\left[\frac{\partial}{\partial y_1}\left\{P_3(y_1,y_2)+P_2(y_1)\right\}\right]\exp\left\{f_2'(y_2)+P_3(y_1,y_2)+P_2(y_1)\right\}\nonumber\\
&\pushright{+\left\{\frac{\partial}{\partial y_1}P_3(y_1,y_2)\right\}\exp\left\{f_2'(y_2)+P_3(y_1,y_2)\right\},}
\end{align*}
which simplifies to 
\begin{multline*}
\left[\frac{\partial }{\partial y_1}\left\{f_1(y_1)+P_1(y_1,y_2)-P_3(y_1,y_2)-P_2(y_1)\right\}\right]\\
\times \exp\left\{f_2'(y_2)+P_3(y_1,y_2)+P_2(y_1)\right\}
\\+\left[\frac{\partial }{\partial y_1}\left\{f_1(y_1)+P_1(y_1,y_2)-P_3(y_1,y_2)\right\}\right]\\
\times \exp\left\{f_2'(y_2)+P_3(y_1,y_2)\right\}=0.
\end{multline*}
Since $\exp\left\{f_2'(y_2)+P_3(y_1,y_2)\right\}\neq 0$,
this becomes 
\begin{multline}\label{eq_cross_mult_diff_general2}
\left[\frac{\partial }{\partial y_1}\left\{f_1(y_1)+P_1(y_1,y_2)-P_3(y_1,y_2)-P_2(y_1)\right\}\right]\exp\left\{P_2(y_1)\right\}\\
+\left[\frac{\partial }{\partial y_1}\left\{f_1(y_1)+P_1(y_1,y_2)-P_3(y_1,y_2)\right\}\right]=0.\end{multline}
Focusing on the components that involve $y_2$, we see that 
\[
\left[\frac{\partial }{\partial y_1}\left\{P_1(y_1,y_2)-P_3(y_1,y_2)\right\}\right]\left[\exp\left\{P_2(y_1)\right\}+1\right]\]
does not depend on $y_2$. 
Since $(\exp(P_2(y_1))+1)>0$, we have 
\[\frac{\partial^2}{\partial y_1\partial y_2}\left\{P_1(y_1,y_2)-P_3(y_1,y_2)\right\}=0.\] Recall that 
\begin{equation}\label{eq_P3P1diff}
P_1(y_1,y_2)-P_3(y_1,y_2)=\alpha_2(y_1)+y_2\beta_2(y_1)-y_2^2k_2/2-\alpha_1'(y_2)-y_1\beta_1'(y_2)+y_1^2k_1'/2.
\end{equation}
So $0=\frac{\partial^2}{\partial y_1\partial y_2}\{P_1(y_1,y_2)-P_3(y_1,y_2)\}=\frac{\d \beta_2(y_1)}{\d y_1}-\frac{\d \beta_1'(y_2)}{\d y_2}$ implies that $\beta_2$ and $\beta_1'$ are both linear with the same coefficient on the linear term. Now that $\beta_2$ has the form $\beta_2(y_1)=c_{\beta_2,-1}+c_{\beta_2,0}\mathds{1}_{y_1}+c_{\beta_2,1}y_1$, write $\beta_{2;-1,0}\equiv c_{\beta_2,-1}+c_{\beta_2,0}=\beta_2(0)+c_{\beta_2,0}$ as a shorthand notation for $\beta_2$ with indicator set to $1$ while $y_1$ set to $0$. Similarly define $\beta'_{1;-1,0}\equiv c_{\beta_1',-1}+c_{\beta_1',0}=\beta_1'(0)+c_{\beta_1',0}$. Then for $y_1,y_2\neq 0$ since $c_{\beta_2,1}=c_{\beta_1',1}$, we necessarily have 
\begin{align*}
y_2\beta_2(y_1)-y_1\beta_1'(y_2)&=y_2(c_{\beta_2,-1}+c_{\beta_2,0}+c_{\beta_2,1}y_1)-y_1(c_{\beta_1',-1}+c_{\beta_1',0}+c_{\beta_1',1}y_2) \\
&= y_2\beta_{2;-1,0}-y_1\beta'_{1;-1,0},
\end{align*}
and so by (\ref{eq_P3P1diff})
\begin{align*}
&\,P_1(y_1,y_2)-P_3(y_1,y_2)\\
=&\,\left(\alpha_2(y_1)-y_1\beta'_{1;-1,0}+y_1^2k_1'/2\right)-\left(\alpha_1'(y_2)-y_2\beta_{2;-1,0}+y_2^2k_2/2\right)\\
\equiv&\,P_{1,3}(y_1)-(\text{function in }y_2\text{ only}).
\end{align*} 
Plugging this into (\ref{eq_cross_mult_diff_general2}), we get
\[\left[\frac{\d }{\d y_1}\left\{f_1(y_1)+P_{1,3}(y_1)-P_2(y_1)\right\}\right]\exp\left\{P_2(y_1)\right\}
+\left[\frac{\d }{\d y_1}\left\{f_1(y_1)+P_{1,3}(y_1)\right\}\right]\]
equals 0, or equivalently
\[\left[\frac{\d }{\d y_1}\left\{f_1(y_1)+P_{1,3}(y_1)\right\}\right]\left[\exp\{P_2(y_1)\}+1\right]=\left\{\frac{\d}{\d y_1}P_2(y_1)\right\}\exp\left\{P_2(y_1)\right\}.\]
Then
\begin{align*}
f_1(y_1)&=\int \frac{\exp\{P_2(y_1)\}\{\d P_2(y_1)/\d y_1\}}{\exp\{P_2(y_1)\}+1}\d y_1-P_{1,3}(y_1)\\
&=\log\left[1+\exp\{P_2(y_1)\}\right]-P_{1,3}(y_1)+\mathrm{const}.
\end{align*}
So for $y_1\neq 0$,
\begin{align}
\exp(f_1(y_1))
\propto&\,\frac{1+\exp\{P_2(y_1)\}}{\exp\{P_{1,3}(y_1)\}}\nonumber\\
=&\,\frac{1+\sqrt{2\pi/k_2}\exp\{\alpha_2(y_1)+\beta_2(y_1)^2/(2k_2)\}}{\exp\{\alpha_2(y_1)-\beta'_{1;-1,0}y_1+y_1^2k_1'/2\}}\nonumber\\
=&\,\exp\{-\alpha_2(y_1)+y_1\beta'_{1;-1,0}-y_1^2k_1'/2\}\nonumber\\
&\pushright{+\sqrt{2\pi/k_2}\exp\{y_1\beta'_{1;-1,0}+\beta_2(y_1)^2/(2k_2)-y_1^2k_1'/2\}.}\label{f1_form1}
\end{align}

Now condition on the event $\mathds{1}_{y_1}=1$ and $\mathds{1}_{y_2}=0$. Then (\ref{eq_joint_dist_equal_general}) becomes
\[\frac{\exp\{f_1(y_1)\}}{\sqrt{2\pi/k_2}\exp\{\alpha_2(y_1)+\beta_2(y_1)^2/(2k_2)\}+1}\mathds{1}_{y_1}
\propto\exp\{y_1\beta'_1(0)-y_1^2k_1'/2\}\mathds{1}_{y_1},\]
which implies that  for $y_1\neq 0$,
\begin{multline}\label{f1_form2}
\exp\left\{f_1(y_1)\right\}\propto \exp\{y_1\beta_1'(0)-y_1^2k_1'/2\}\\
+\sqrt{2\pi/k_2}\exp\left\{y_1\beta_1'(0)-y_1^2k_1'/2+\alpha_2(y_1)+\beta_2(y_1)^2/(2k_2)\right\}.
\end{multline}
Applying Lemma \ref{lem_match} to (\ref{f1_form1}) and (\ref{f1_form2}), by matching the terms we have (conditional on $y_1\neq 0$) either
\begin{align}
-\alpha_2(y_1)+y_1\beta'_{1;-1,0}&=y_1\beta_1'(0)+\mathrm{const};\quad\mathrm{or}\label{eq_matched_case1}\\
-\alpha_2(y_1)+y_1\beta'_{1;-1,0}&=y_1\beta_1'(0)+\alpha_2(y_1)+\beta_2(y_1)^2/(2k_2)+\mathrm{const}\quad\mathrm{and}\nonumber\\
y_1\beta'_{1;-1,0}+\beta_2(y_1)^2/(2k_2)&= y_1\beta_1'(0)+\mathrm{const}.\label{eq_matched_case2}
\end{align}
Conditional on $y_1\neq 0$, in the first case (\ref{eq_matched_case1}), $\alpha_2(y_1)=y_1c_{\beta_1',0}+\mathrm{const}$; in the second case (\ref{eq_matched_case2}), $\alpha_2(y_1)+\beta_2(y_1)^2/(2k_2)=\mathrm{const}$ and $\beta_2(y_1)^2/(2k_2)=-y_1c_{\beta_1',0}+\mathrm{const}$, which implies $\beta_2(y_1)=\mathrm{const}$ and $\alpha_2(y_1)=\mathrm{const}$ for $y_1\neq 0$, and $c_{\beta_1',0}=0$, which in turn implies (\ref{eq_matched_case1}). Thus, in either case, $\alpha_2(y_1)=c_{\alpha_2,0}\mathds{1}_{y_1}+y_1c_{\beta_1',0}+\mathrm{const}$, i.e.~$\alpha_2$ is linear (or constant) in $y_1\neq 0$ with coefficient on $y_1$ equal to $c_{\beta_1',0}$. By (\ref{eq_matched_case1}) for $y_1\neq 0$,
\begin{multline}\label{f1_form3}
\exp\{f_1(y_1)\}\propto\exp\{y_1\beta_1'(0)-y_1^2k_1'/2\}\\
+\sqrt{2\pi/k_2}\exp\{y_1\beta'_{1;-1,0}+\beta_2(y_1)^2/(2k_2)-y_1^2k_1'/2\},
\end{multline}
clearly a single univariate Gaussian or a mixture of two univariate Gaussian distributions (since $\beta_2$ is at most linear in $y_1$). Similarly, we must have $\alpha_1'(y_2)=y_2\beta_{2;-1,0}-y_2\beta_2(0)+\mathrm{const}=y_2c_{\beta_2,0}+\mathrm{const}$ for $y_2\neq 0$, and for $y_2\neq 0$
\begin{multline}\label{f2_form3}
\exp\{f_2'(y_2)\}\propto\exp\{y_2\beta_2(0)-y_2^2k_2/2\}\\
+\sqrt{2\pi/k_1'}\exp\{y_2\beta_{2;-1,0}+\beta_1'(y_2)^2/(2k_1')-y_2^2k_2/2\}.
\end{multline}

Now suppose by contradiction that $\exp\left\{f_1(y_1)\right\}$ given $y_1\neq 0$ has only one Gaussian component, instead of being a sum of two Gaussian densities. Then by (\ref{f1_form3}), $\beta_1'(0)=\beta'_{1;-1,0}$ and $\beta_2(y_1)$ is a constant given $\mathds{1}_{y_1}$, i.e.~$\beta_2(y_1)=c_{\beta_2,-1}+c_{\beta_2,0}=\beta_{2;-1,0}$ for $y_1\neq 0$. Plugging this into the left-hand side of (\ref{eq_joint_dist_equal_general}) and integrating w.r.t.~$\lambda(y_1)$, the continuous part ($y_2\neq 0$) of the marginal distribution of $y_2$ given $\boldsymbol{Y}_\mathcal{P}\equiv\boldsymbol{y}_\mathcal{P} $ is 
\begin{multline*}
\exp\{f_2'(y_2)\}\propto\frac{\exp\{f_1(0)+\alpha_2(0)+y_2\beta_2(0)-y_2^2k_2/2\}}{\sqrt{2\pi/k_2}\exp\{\alpha_2(0)+\beta_2(0)^2/(2k_2)\}+1}\\
+\exp\{y_2\beta_{2;-1,0}-y_2^2k_2/2\}\int_{\mathbb{R}}\frac{\exp\{\delta_{1}+f_1(y_1)+\alpha_2(y_1)\}}{\sqrt{2\pi/k_2}\exp\{\alpha_2(y_1)+\beta_2(y_1)^2/(2k_2)\}+1}\d y_1,
\end{multline*}
which is a mixture between $\mathcal{N}(\beta_2(0)/k_2,1/k_2)$ and $\mathcal{N}(\beta_{2;-1,0}/k_2,1/k_2)$, i.e.~the variance in both components are equal. Note that the integral in the second term is a Lebesgue integral. This together with (\ref{f2_form3}) implies that $\beta_1'(y_2)$ cannot depend on the value of $y_2$ given $y_2\neq 0$, i.e.~$\beta_1'(y_2)=c_{\beta_1',-1}+c_{\beta_1',0}=\beta'_{1;-1,0}$. Since we already know that $\beta_1'(0)=\beta'_{1;-1,0}$ by discussion above, this implies that $\beta_1'$ is an absolute constant in $y_2$ and $\mathds{1}_{y_2}$, and also that $\alpha_2$ may depend on $y_1$ only through $\mathds{1}_{y_1}$, a contradiction to the assumption of the theorem.

Thus, (\ref{f1_form3}) and (\ref{f2_form3}) will both have to be mixtures of precisely two Gaussians, and so by definition the joint distribution $p(\boldsymbol{Y})$ of $\boldsymbol{Y}$ must be of 2-Gaussian type with respect to $\mathcal{G}$ and $\mathcal{G}'$.
\end{proof}

\begin{proof}[Proof of Corollary \ref{id_binary}]
When $|\mathcal{V}|=2$, in Proposition \ref{prop_peters} we always have $\mathcal{P}=\varnothing$ and $V_1$ does not have a parent in $\mathcal{G}$, so $P(Y_{V_1}=y|Y_{V_1}\neq 0)$ by definition is  just a Gaussian, not a mixture two Gaussians, and hence $p(\boldsymbol{Y})$ cannot be of 2-Gaussian type with respect to any pairs of distinct Markov equivalent graphs.

Now consider $|\mathcal{V}|=3$, and assume the two vertices with reversible edges in Proposition \ref{prop_peters} are $V_1$ and $V_2$, and that $V_1\to V_2$ in $\mathcal{G}$ and $V_1\leftarrow V_2$ in $\mathcal{G}'$. If neither $V_1$ or $V_2$ has $V_3$ as its parent in both graphs, then we can marginalize $V_3$ out and it reduces to the 2-d case. Suppose otherwise. Then we must have (1) $V_1\to V_2\leftarrow V_3$ in $\mathcal{G}$, or (2) $V_2\to V_1\leftarrow V_3$ in $\mathcal{G}'$, or (3) an additional edge between $V_1$ and $V_3$ added to (1), or (4) an additional edge between $V_2$ and $V_3$ added to (2). 

For (1) and (2) both graphs are the only graph in their Markov equivalence class; for (3) the reversible edge becomes $V_1$---$V_3$ violating the assumption (and in fact one can marginalize out the common child $V_2$ and get back to the 2-d case), and similarly for (4). Thus, we have again ruled out the possibility of any pair of distinct Markov equivalent graphs with respect to which $p(\boldsymbol{Y})$ can be of 2-Gaussian type.
\end{proof}

\begin{remark}\label{alpha_relationships}
In the proof of Theorem \ref{id_full}, we proved that whenever $p(\boldsymbol{Y})$ factorizes with respect to two distinct graphs $\mathcal{G}$ and $\mathcal{G}'$ (whenever identifiability does not hold), everything up to (\ref{f2_form3}) in the proof must hold. Specifically, conditioning on almost every $\boldsymbol{y}_{\mathcal{P}}$, $\alpha_2$ and $\beta_2$ in $\mathcal{G}$ as well as $\alpha_1'$ and $\beta_1'$ in $\mathcal{G}'$ can be at most linear in $y_1$ and $y_2$, respectively, namely
\begin{align*}
\beta_1'(y_2)&=c_{\beta_1',-1}+c_{\beta_1',0}\mathds{1}_{y_2}+c_{\beta_1',1}y_2,\quad \beta_2(y_1)=c_{\beta_2,-1}+c_{\beta_2,0}\mathds{1}_{y_1}+c_{\beta_2,1}y_1,\\
\alpha_1'(y_2)&=c_{\alpha'_1,-1}+c_{\alpha'_1,0}\mathds{1}_{y_2}+c_{\alpha'_1,1}y_2,\quad\alpha_2(y_1)=c_{\alpha_2,-1}+c_{\alpha_2,0}\mathds{1}_{y_1}+c_{\alpha_2,1}y_1,
\end{align*}
with coefficients depending on $\boldsymbol{y}_{\mathcal{P}}$ where
\begin{align}\label{eq_coef_relationships}
c_{\alpha'_1,1}=c_{\beta_2,0},\quad c_{\alpha_2,1}=c_{\beta'_1,0},\quad c_{\beta_1',1}=c_{\beta_2,1}.
\end{align}
It is noted that, although not used in deriving our conclusion involving $2$-Gaussian type distributions, we in addition also have the following results.
\begin{equation*}
c_{\alpha_1',-1}=c_{\alpha_2,-1},\quad c_{\alpha_1',0}=c_{\alpha_2,0},\quad c_{\alpha_1',-1}+c_{\alpha_1',0}=c_{\alpha_2,-1}+c_{\alpha_2,0}=0.
\end{equation*}
These might shed some light on how to show that distributions of 2-Gaussian type do not exist for a general $m\geq 4$.
\end{remark}

\begin{proof}[Proof of Remark \ref{alpha_relationships}]

By (\ref{eq_joint_dist_equal_general}), (\ref{f1_form3}), (\ref{f2_form3}), the joint distribution of $Y_1$ and $Y_2$ conditional on $\boldsymbol{Y}_{\mathcal{P}}$ has two characterizations (up to normalizing constants)
\begin{align}
&\,\frac{\exp\{\mathds{1}_{y_1}\delta_{1}+y_1\beta_1'(0)-y_1^2k_1'/2+\mathds{1}_{y_2}\alpha_2(y_1)+y_2\beta_2(y_1)-y_2^2k_2/2\}}{\sqrt{2\pi/k_2}\exp\left\{\alpha_2(y_1)+\beta_2(y_1)^2/(2k_2)\right\}+1}\nonumber\\
+&\,\frac{\sqrt{2\pi/k_2}\exp\{\mathds{1}_{y_1}\delta_{1}+y_1\beta'_{1;-1,0}+\beta_2(y_1)^2/(2k_2)-y_1^2k_1'/2+\mathds{1}_{y_2}\alpha_2(y_1)+y_2\beta_2(y_1)-y_2^2k_2/2\}}{\sqrt{2\pi/k_2}\exp\left\{\alpha_2(y_1)+\beta_2(y_1)^2/(2k_2)\right\}+1}\nonumber\\
\propto&\,\frac{\exp\{\mathds{1}_{y_2}\delta_2'+y_2\beta_2(0)-y_2^2k_2/2+\mathds{1}_{y_1}\alpha_1'(y_2)+y_1\beta_1'(y_2)-y_1^2k_1'/2\}}{\sqrt{2\pi/k_1'}\exp\left\{\alpha_1'(y_2)+\beta_1'(y_2)^2/(2k_1')\right\}+1}\nonumber\\
+&\,\frac{\sqrt{2\pi/k_1'}\exp\{\mathds{1}_{y_2}\delta_2'+y_2\beta_{2;-1,0}+\beta_1'(y_2)^2/(2k_1')-y_2^2k_2/2+\mathds{1}_{y_1}\alpha_1'(y_2)+y_1\beta_1'(y_2)-y_1^2k_1'/2\}}{\sqrt{2\pi/k_1'}\exp\left\{\alpha_1'(y_2)+\beta_1'(y_2)^2/(2k_1')\right\}+1}.\label{eq_joint_full}
\end{align}
Divide both sides by $\exp(y_1\beta_1'(0)+y_2\beta_2(0)-y_1^2k_1'/2-y_2^2k_2/2)$ and expanding $\beta_1'(y_2)$ and $\beta_2(y_1)$, this becomes
\begin{align*}
&\,\frac{\exp\{\mathds{1}_{y_1}\delta_{1}+\mathds{1}_{y_2}\alpha_2(y_1)+y_2c_{\beta_2,0}\mathds{1}_{y_1}+y_1y_2c_{\beta_2,1}\}}{\sqrt{2\pi/k_2}\exp\left\{\alpha_2(y_1)+\beta_2(y_1)^2/(2k_2)\right\}+1}\\
+&\,\frac{\sqrt{2\pi/k_2}\exp\{\mathds{1}_{y_1}\delta_{1}+y_1c_{\beta'_1,0}+\beta_2(y_1)^2/(2k_2)+\mathds{1}_{y_2}\alpha_2(y_1)+y_2c_{\beta_2,0}\mathds{1}_{y_1}+y_1y_2c_{\beta_2,1}\}}{\sqrt{2\pi/k_2}\exp\left\{\alpha_2(y_1)+\beta_2(y_1)^2/(2k_2)\right\}+1}\\
\\
\propto&\,\frac{\exp\{\mathds{1}_{y_2}\delta_2'+\mathds{1}_{y_1}\alpha_1'(y_2)+y_1c_{\beta_1',0}\mathds{1}_{y_2}+y_1y_2c_{\beta_1',1}\}}{\sqrt{2\pi/k_1'}\exp\left\{\alpha_1'(y_2)+\beta_1'(y_2)^2/(2k_1')\right\}+1}\\
+&\,\frac{\sqrt{2\pi/k_1'}\exp\{\mathds{1}_{y_2}\delta_2'+y_2c_{\beta_{2},0}+\beta_1'(y_2)^2/(2k_1')+\mathds{1}_{y_1}\alpha_1'(y_2)+y_1c_{\beta_1',0}\mathds{1}_{y_2}+y_1y_2c_{\beta_1',1}\}}{\sqrt{2\pi/k_1'}\exp\left\{\alpha_1'(y_2)+\beta_1'(y_2)^2/(2k_1')\right\}+1}.
\end{align*}

Now expanding $\alpha_1'(y_2)$  and $\alpha_2(y_1)$ and using the relationships in (\ref{eq_coef_relationships}), we divide both sides by $\exp(y_{1}c_{\alpha_2,1}\mathds{1}_{y_2}+y_2c_{\beta_2,0}\mathds{1}_{y_1}+y_1y_2c_{\beta_2,1})=\exp(y_{1}c_{\beta_1',0}\mathds{1}_{y_2}+y_2c_{\alpha_1',1}\mathds{1}_{y_1}+y_1y_2c_{\beta_2,1})$ and get
\begin{align}
&\,\frac{\exp\{\mathds{1}_{y_1}\delta_{1}+\mathds{1}_{y_2}(c_{\alpha_2,-1}+c_{\alpha_2,0}\mathds{1}_{y_1})\}}{\sqrt{2\pi/k_2}\exp\left\{\alpha_2(y_1)+\beta_2(y_1)^2/(2k_2)\right\}+1}\nonumber\\
+&\,\frac{\sqrt{2\pi/k_2}\exp\left\{\mathds{1}_{y_1}\delta_{1}+y_1c_{\beta'_1,0}+\beta_2(y_1)^2/(2k_2)+\mathds{1}_{y_2}(c_{\alpha_2,-1}+c_{\alpha_2,0}\mathds{1}_{y_1})\right\}}{\sqrt{2\pi/k_2}\exp\left\{\alpha_2(y_1)+\beta_2(y_1)^2/(2k_2)\right\}+1}\nonumber
\\
=&\,C_0\frac{\exp\{\mathds{1}_{y_2}\delta_2'+\mathds{1}_{y_1}(c_{\alpha_1',-1}+c_{\alpha_1',0}\mathds{1}_{y_2})\}}{\sqrt{2\pi/k_1'}\exp\left\{\alpha_1'(y_2)+\beta_1'(y_2)^2/(2k_1')\right\}+1}\nonumber\\
+&\,C_0\frac{\sqrt{2\pi/k_1'}\exp\{\mathds{1}_{y_2}\delta_2'+y_2c_{\beta_{2},0}+\beta_1'(y_2)^2/(2k_1')+\mathds{1}_{y_1}(c_{\alpha_1',-1}+c_{\alpha_1',0}\mathds{1}_{y_2})\}}{\sqrt{2\pi/k_1'}\exp\left\{\alpha_1'(y_2)+\beta_1'(y_2)^2/(2k_1')\right\}+1} \label{eq_full_simplified}
\end{align}
for some $C_0$. Setting $\mathds{1}_{y_1}=\mathds{1}_{y_2}=0$ (\ref{eq_full_simplified}) becomes
\begin{multline}\label{eq_full_simplified_00}
\frac{1+\sqrt{2\pi/k_2}\exp\{c_{\beta_2,-1}^2/(2k_2)\}}{\sqrt{2\pi/k_2}\exp\{c_{\alpha_2,-1}+c_{\beta_2,-1}^2/(2k_2)\}+1}=C_0\frac{1+\sqrt{2\pi/k_1'}\exp\{c_{\beta_1'1,-1}^2/(2k_1')\}}{\sqrt{2\pi/k_1'}\exp\{c_{\alpha_1',-1}+c_{\beta_1',-1}^2/(2k_1')\}+1},
\end{multline}
and with $\mathds{1}_{y_1}\neq 0$, $\mathds{1}_{y_2}=0$ (\ref{eq_full_simplified}) becomes
\begin{multline}\label{eq_full_simplified_10}
\exp(\delta_1)\frac{1+\sqrt{2\pi/k_2}\exp\{y_1c_{\beta_1',0}+\beta_2(y_1)^2/(2k_2)\}}{\sqrt{2\pi/k_2}\exp\{c_{\alpha_2,-1}+c_{\alpha_2,0}+c_{\alpha_2,1}y_1+\beta_2(y_1)^2/(2k_2)\}+1}\\
=C_0\exp(c_{\alpha_1',-1})\frac{1+\sqrt{2\pi/k_1'}\exp\{c_{\beta_1'1,-1}^2/(2k_1')\}}{\sqrt{2\pi/k_1'}\exp\{c_{\alpha_1',-1}+c_{\beta_1',-1}^2/(2k_1')\}+1}.
\end{multline}
Since the right-hand side of (\ref{eq_full_simplified_10}) is a constant, by matching the numerator and the denominator of the left-hand side using Lemma \ref{lem_match}, we must have either (i) $y_1c_{\beta_1',0}+\beta_2(y_1)^2/(2k_2)=c_{\alpha_2,-1}+c_{\alpha_2,0}+c_{\alpha_2,1}y_1+\beta_2(y_1)^2/(2k_2)$, or (ii) $y_1c_{\beta_1',0}+\beta_2(y_1)^2/(2k_2)=\mathrm{const}$ for $y_1\neq 0$. But (ii) implies that $c_{\beta_2,1}=c_{\beta_1',0}=0$, which by $c_{\beta_1',1}=c_{\beta_2,1}$ implies that $\beta_1'$ is an absolute constant in $y_2\in\mathbb{R}$, a violation to the assumption. Thus (i) holds, and by $c_{\beta_1',0}=c_{\alpha_2,1}$ this implies that 
\begin{equation}\label{eq_alphas_relationship}
\alpha_{2;-1,0}\equiv c_{\alpha_2,-1}+c_{\alpha_2,0}=0,\quad\text{and by symmetry}\quad \alpha_{1;-1,0}'\equiv c_{\alpha_1',-1}+c_{\alpha_1',0}=0.
\end{equation}
Thus the left-hand side of (\ref{eq_full_simplified_10}) is just $\exp(\delta_1)$. Note that the right-hand side of (\ref{eq_full_simplified_10}) is $\exp(c_{\alpha_1',-1})$ times the right-hand side of (\ref{eq_full_simplified_00}). So by equating the left-hand side of (\ref{eq_full_simplified_10}) with $\exp(c_{\alpha_1',-1})$ times the left-hand side of (\ref{eq_full_simplified_00}) we have
\begin{equation}\label{eq_delta1}\exp(\delta_1)=\exp\left(c_{\alpha_1,-1}'\right)\frac{1+\sqrt{2\pi/k_2}\exp\left\{c_{\beta_2,-1}^2/(2k_2)\right\}}{\sqrt{2\pi/k_2}\exp\left\{c_{\alpha_2,-1}+c_{\beta_2,-1}^2/(2k_2)\right\}+1)}
\end{equation}
and similarly 
\begin{equation}\label{eq_delta2}
\exp(\delta_2')=\exp\left(c_{\alpha_2,-1}\right)\frac{1+\sqrt{2\pi/k_1'}\exp\left\{c_{\beta_1'1,-1}^2/(2k_1')\right\}}{\sqrt{2\pi/k_1'}\exp\left\{c_{\alpha_1',-1}+c_{\beta_1',-1}^2/(2k_1')\right\}+1}.
\end{equation}
Now by (\ref{eq_alphas_relationship}), with $\mathds{1}_{y_1}=\mathds{1}_{y_2}=1$, (\ref{eq_full_simplified}) simplifies to
$\exp(\delta_1)=C_0\cdot\exp(\delta_2')$. Thus by (\ref{eq_full_simplified_00}), (\ref{eq_delta1}) and (\ref{eq_delta2}), one get
\[C_0=\frac{\exp(\delta_1)}{\exp(\delta_2')}=\frac{\exp(c_{\alpha_1,-1}')\frac{1+\sqrt{2\pi/k_2}\exp\left\{c_{\beta_2,-1}^2/(2k_2)\right\}}{\sqrt{2\pi/k_2}\exp\left\{c_{\alpha_2,-1}+c_{\beta_2,-1}^2/(2k_2))+1\right\}}}{\exp(c_{\alpha_2,-1})\frac{1+\sqrt{2\pi/k_1'}\exp\left\{c_{\beta_1'1,-1}^2/(2k_1')\right\}}{\sqrt{2\pi/k_1'}\exp\left\{c_{\alpha_1',-1}+c_{\beta_1',-1}^2/(2k_1')\right\}+1}}=\frac{\exp(c_{\alpha_1,-1}')}{\exp(c_{\alpha_2,-1})}C_0\]
and thus $c_{\alpha_1,-1}'=c_{\alpha_2,-1}$. Combining with (\ref{eq_alphas_relationship}), we get
\begin{equation}\label{eq_alphas_relationship2}
c_{\alpha_1',-1}=c_{\alpha_2,-1},\quad c_{\alpha_1',0}=c_{\alpha_2,0},\quad c_{\alpha_1',-1}+c_{\alpha_1',0}=c_{\alpha_2,-1}+c_{\alpha_2,0}=0.
\end{equation}
Note that this result holds as long as we assume identifiability does not hold.

\end{proof}

\end{document}